\documentclass[submission, Phys]{SciPost}
\usepackage{verbatim}
\usepackage{hyperref}
\def\doi{http://dx.doi.org/}
\usepackage{amsmath,braket,mathdots,mathtools,amssymb,xcolor,calligra,color}
\usepackage{amsthm} 
\usepackage{bbm}  
\usepackage{fancyhdr}
\usepackage{float}
\usepackage{bm}
\usepackage{authblk}
\usepackage{color}
\usepackage{authblk}
\usepackage{cite} 
\usepackage{graphicx}

\usepackage{tikz}
\usetikzlibrary{spy}
\usetikzlibrary{patterns}
\usetikzlibrary{calc,arrows,positioning}
\usetikzlibrary{arrows,shadows}
\usetikzlibrary{decorations.pathmorphing}	
\usetikzlibrary{decorations.markings}
\usepackage{pgfplots}

\usepackage{pgfplotstable}
\usepackage{filecontents}
\usepackage[utf8]{inputenc}    

\newcommand{\be}{\begin{equation}}
\newcommand{\ee}{\end{equation}}
\newcommand{\bea}{\begin{eqnarray}}
\newcommand{\eea}{\end{eqnarray}}

\def\tr{\text{tr}\,}

\def\XXint#1#2#3{{\setbox0=\hbox{$#1{#2#3}{\int}$}
     \vcenter{\hbox{$#2#3$}}\kern-.5\wd0}}

\def\nn{\nonumber\\}

\def\fr#1{(\ref{#1})}
\def\ontop#1#2{\genfrac{}{}{0pt}{}{#1}{#2}}

\let\OLDthebibliography\thebibliography
\renewcommand\thebibliography[1]{
  \OLDthebibliography{#1}
  \setlength{\parskip}{0pt}
  \setlength{\itemsep}{0pt plus 0.3ex}
}

\newcommand{\sign}{\,\text{sgn}\,}

\newcommand{\pf}{\,{\rm pf}}
\newcommand{\Pf}{\,{\rm Pf}}
\newcommand{\Det}{\,{\rm Det}}

\newcommand{\1}{\,\pmb{1}}
\newcommand{\D}[1]{\text{d}#1}

\newtheorem{theorem}{Theorem}

\newtheorem{property}{Lemma}

\hypersetup{
     colorlinks=true,
     linkcolor=blue,
     filecolor=blue,
     citecolor = green,      
     urlcolor=cyan,
     }

\begin{document}
\begin{center}
{\Large\bf Out-of-equilibrium dynamics of the XY spin chain\\ from form factor expansion}
\end{center}
\begin{center}
Etienne Granet\textsuperscript{1$\star$}, Henrik Dreyer\textsuperscript{1} and Fabian H.L. Essler\textsuperscript{1},
\end{center}
\begin{center}
{\bf 1} The Rudolf Peierls Centre for Theoretical Physics, Oxford
University, Oxford OX1 3PU, UK\\
${}^\star$ {\small \sf etienne.granet@physics.ox.ac.uk}
\end{center}
\date{\today}

\section*{Abstract}
{\bf We consider the XY spin chain with arbitrary time-dependent
  magnetic field and anisotropy. We argue that a certain subclass of
  Gaussian states, called Coherent Ensemble (CE) following \cite{DBG},
  provides a natural and unified framework for out-of-equilibrium
  physics in this model. We show  that \textit{all} correlation
  functions in the CE can be computed using form factor expansion
  and expressed in terms of Fredholm determinants. In particular, we
  present exact out-of-equilibrium expressions in the thermodynamic 
  limit for the previously unknown order parameter $1$-point function,
  dynamical $2$-point function and equal-time $3$-point function. 
}

\tableofcontents

%

\renewcommand\Affilfont{\fontsize{9}{10.8}\itshape}

\section{Introduction}
Quantum integrable models are special models of many-body quantum
physics with both a rich phenomenology and an exact Bethe-ansatz
solution. But despite their ``exact solvability", obtaining
closed-form expressions in the thermodynamic limit for correlations of
local observables in and out-of-equilibrium remains a formidable
challenge. The standard approach to these problems \cite{vladbook}
consists in expressing such correlation functions as
form factor sums over the full Hilbert space. In interacting
models, this task has been achieved only in certain parameter
regimes, such as ground state correlations at late times and large
  distances \cite{kitanineetalg,kitanineetcformfactor,kozlowskimaillet,kozlowski4},
equal-time finite temperature correlations at short or large distances
\cite{suzuki85,klumper92,patuklumper,kozlowskimailletslvanov,kozlowskimailletslvanov2,Damerau07,Boos08,Trippe10,Dugave13,Dugave14},
full correlations in systematic strong coupling expansions
  \cite{GE20,GE21} or expansions in low densities of excitations \cite{EK09,PT10,G21}, and also in some 
particularly simple interacting models \cite{pozsgay,pozsgay2,solvable}. 
A number of numerical, approximate, field theory and other approaches
aimed at facilitating form factor summations have been developed
over the last decade and a half \cite{CC06,cauxcalabreseslavnov,EK08,CE13,PC14,deNP15,deNP16,CADY16,BCDF16,Doyon18,DNP18,panfil20,cortescuberopanfil1,cortescuberopanfil2,kozlowski1,CCYS21,NBMP21}.  

A subclass of quantum integrable models has arguably been of
particular importance, namely theories that can be formulated in
terms of free fermions. Examples include the Lieb-Liniger model
at infinite coupling \cite{LiebLiniger63a,girardeau} and the XY model
in a field \cite{xymodel1,BMD70,xymodel2,BM71,xymodel3,xymodel4}. They
constitute the point of departure and testbed of any field theory or
exact method applying to the interacting case. But despite their free
fermion formulation, the problem of obtaining analytic expressions for
general in- and out-of-equilibrium correlations in the thermodynamic
limit is still unsolved for some of these models. In fact, the
computation of in- and out-of-equilibrium correlations can be said to
be ``fully'' solved only for the models with a $U(1)$ symmetry
such as the Lieb-Liniger model at infinite coupling and the XX
chain. In this case, there exist integral, Pfaffian or Fredholm
determinant representations for all static and dynamical correlations
in arbitrary eigenstates
\cite{KS90,slavnov90,izergin87,Its93,kojima97,niemeijer,colomo,xxcorr1,leclair1996,GKS18},
as well as for the full out-of-equilibrium time evolution of
correlations after quantum quenches \cite{niemeijer,denardis15}. The 
exact tractability of the form factor expansion in these cases
originates from the Cauchy determinant structure of the form factors
\cite{KS90}. 

However, there are still unknown correlation functions in the
thermodynamic limit of free fermionic models
without $U(1)$ symmetry such as the Transverse Field Ising Model and
more generally the XY model in a field, despite a vast literature on
the subject, see e.g.
\cite{xymodel1,xymodel2,xymodel3,Perk80,Perk84,Perk09,MPS83a,MPS83b,izergin2000,doyongamsa,Kapitonov03,SPS:04,IJK05,Cherng06,cd-07,RSMS09,RSMS10,CEF1,CEF2,CEF3,FE13,SE12,essler12,Foini11,Foini12,BRI12,e-13,ia-13,klich14,Schuricht16,Groha18,Steinberg,GFE20,GLC20,GIZ20,Hodsagi20,ares}. While
the quantities that are local in the underlying fermions (such as any
correlation of the transverse magnetization, and the static $2n$-point
functions of the order parameter) can be computed efficiently with
Pfaffian representations arising from Wick's theorem\cite{xymodel2,xymodel3},
there are no known exact representations for expectation values in
general Gaussian states in the thermodynamic limit for the quantities
that are non-local in the underlying fermions (such as static
$(2n+1)$-point functions of the order parameter, or any dynamical
correlation of the order parameter). What makes the form factor
expansion difficult to compute in these cases despite the model being
free is that the form factors of the order parameter are not of Cauchy
form. As a consequence, in terms of difficulty of the calculation
these free models without $U(1)$ symmetry can be considered as in
certain ways intermediate cases between free $U(1)$-symmetric and
interacting models.

In contrast to the situation in the thermodynamic limit 
Pfaffian representations are readily available in finite
systems with open boundary conditions, see
e.g. \cite{derzkho}. Similarly, for finite systems with periodic
boundary conditions one can invoke clustering properties
\cite{xymodel3,xymodel4,CEF2,Foini12,FE13} to obtain approximate
representations. However, these representations typically scale with
system size and as a far as we are aware their thermodynamic limits
are generally not known. 


In this work we show how to perform the form factor expansion for
expectation values of \textit{arbitrary} operators out of equilibrium. In particular, we derive the full time evolution 
of the order parameter one-point function, dynamical two-point
function and static three-point function under arbitrary
time-dependent ramps of the magnetic field and the
anisotropy. We also derive alternative Fredholm determinant
expressions for the full counting statistics of the transverse
magnetization and order parameter two-point function using form 
factor expansions rather than Wick's theorem, hence with a method that is
more generalizable to interacting models.   
This puts the XY model in a field on the same footing as the models
with $U(1)$ symmetry with regards to out-of-equilibrium physics. 

The technique we use to obtain these results is as follows. We define
the \emph{Coherent Ensemble} (CE) as the expectation value of operators
within coherent states, which are superpositions of all
  zero-momentum pair states, weighted by amplitudes which are parameters of the CE. The
crucial property of these coherent states is that they retain
  their structure when expressed in terms of eigenstates of
  the XY Hamiltonian with different values of magnetic field
$h$ and anisotropy $\gamma$, as observed in \cite{DBG} for the
Ising model. This has two consequences: (i) The time
evolution of the initial state with \textit{any} variation of magnetic
field $h(t)$ and anisotropy $\gamma(t)$ can be written as a coherent
state with a certain amplitude; (ii)
Any correlation function in the CE
can be recast as a correlation function in an elementary (classical)
Hamiltonian such as $-\sum_j \sigma^x_j\sigma^x_{j+1}$ for
$h=0,\gamma=1$ or $-\sum_j \sigma^z_j$ for $h=\infty$. At these values
of parameters, the form factors of the order parameter are exactly
Cauchy determinants, which enables one to use the techniques
developed for $U(1)$ symmetric models and obtain Fredholm
determinant expressions in the thermodynamic limit.   

We note that these coherent states appeared more or less explicitly in
different papers in the literature
\cite{izergin2000,Kapitonov03,Dziarmaga05,RSMS09,CEF2}. Most notably in
\cite{izergin2000,Kapitonov03} they were used to obtain a Fredholm determinant
expression for the two-point function of the order parameter at
equilibrium at finite temperature. But to the best of our knowledge
their utility in deriving Fredholm determinant representations
for generic out-of-equilibrium correlators has not been realized
prior to \cite{DBG} and the present work.

The paper is organized as follows. We start by introducing coherent
states in Section \ref{coherent}, and explain why out-of-equilibrium
physics can be written as a CE. Then in Section \ref{sec:ce} we show
that arbitrary expectation values and correlation functions can be
computed within the CE. Their derivation relies on a number of Lemmas
for form factors and summation formulas that are gathered and proven
in Appendix \ref{lemmas}. Finally, in Section \ref{app} we apply
our results to a number of examples including the Kibble-Zurek
mechanism, Floquet physics and quantum quench physics.

\section{Coherent Ensemble in the XY model\label{coherent}}
\subsection{The XY model in a field}
The Hamiltonian of the XY model on a system of size $L$, in a magnetic field $h$ and with anisotropy $\gamma$ is \cite{xymodel1}
\begin{equation}\label{hamiltonian}
H(h,\gamma)=-\sum_{j=1}^L \frac{1+\gamma}{2}\sigma^x_j \sigma^x_{j+1}+\frac{1-\gamma}{2}\sigma^y_j \sigma^y_{j+1}+h\sigma^z_j\,.
\end{equation}
We impose periodic boundary conditions $L+1\equiv 1$. The
diagonalisation of $H(h,\gamma)$ is reviewed in Appendix
\ref{xy}. The Hamiltonian splits into two sectors $H(h,\gamma)=H^{\rm
  NS}(h,\gamma)\oplus H^{\rm R}(h,\gamma)$ called Neveu-Schwarz (NS)
and Ramond (R) sector respectively
\begin{equation}
H^{\rm NS,R}(h,\gamma)=\sum_{k\in{\rm NS,R}} \varepsilon_{h\gamma}(k)
\left(\alpha_{h\gamma;k}^\dagger
\alpha_{h\gamma;k}-\frac{1}{2}\right)\ ,
\end{equation}
where the fermions $\alpha_{h\gamma;k}$ satisfy canonical anti-commutation relations $\{\alpha_{h\gamma;k},\alpha_{h\gamma;p}^\dagger\}=\delta_{k,p}$. Here, NS and R denote the sets
\begin{equation}
\begin{aligned}
{\rm NS}&=\left\{\frac{2\pi(n+1/2)}{L}, n=-L/2,...,L/2-1 \right\}\\
{\rm R}&=\left\{\frac{2\pi n}{L}, n=-L/2,...,L/2-1 \right\}\,,
\label{NSR}
\end{aligned}
\end{equation}
and $\varepsilon_{h\gamma}(k)$ denotes the energy of mode $k$
\begin{equation}
\varepsilon_{h\gamma}(k)=\begin{cases}
2\sqrt{(h-\cos k)^2+\gamma^2 \sin^2 k} & \text{if }k\neq 0\\
-2(1-h)& \text{if }k= 0
\end{cases}\,.
\end{equation}
In these conventions, $\sigma^z$ (resp. $\sigma^x$) is local (resp. non-local) in the underlying Jordan-Wigner fermions\footnote{We note that compared to the previous paper in Ising \cite{DBG} the notations for $\sigma^x,\sigma^z$ have been switched, to match usual conventions in the quantum quench literature.}, see Appendix \ref{xy}. 

Denoting by $|0\rangle_{h\gamma}^{\rm NS,R}$ the respective vacuum states
annihilated by the $\alpha_{h\gamma;k}$'s in the NS and R
sectors, the eigenstates of the model are then 
\begin{equation}\label{eqeqeq}
\begin{aligned}
|\pmb{k}\rangle_{h,\gamma}&=\alpha_{h\gamma;k_1}^\dagger...\alpha_{h\gamma;k_N}^\dagger|0\rangle_{h\gamma}^{\rm NS}\,,\qquad \pmb{k}\subset {\rm NS}\,,\qquad N\text{ even}\\
|\pmb{k}\rangle_{h,\gamma}&=\alpha_{h\gamma;k_1}^\dagger...\alpha_{h\gamma;k_N}^\dagger|0\rangle_{h\gamma}^{\rm R}\,,\qquad \pmb{k}\subset {\rm R}\,,\qquad N\text{ odd}\,.
\end{aligned}
\end{equation}
In these definitions we choose an ordering such that $k_i<k_j$ if $i<j$ and
$k_i\neq 0,k_j\neq 0$. If $0\in \pmb{k}$ then we choose $k_N=0$. 

For $h>1$, the ground state is $|0\rangle_{h\gamma}^{\rm NS}$. For $0<
h<1$ the two lowest energy states are $|0\rangle_{h\gamma}^{\rm NS}$
and $\alpha_{h\gamma;0}^\dagger|0\rangle_{h\gamma}^{\rm R}$. Their
energy levels are exponentially close in $L$, and in finite size the
true ground state is $|0\rangle_{h\gamma}^{\rm NS}$. The model has two
critical lines $|h|=1,\gamma\neq 0$, and for $\gamma=0, |h|<1$ \cite{xymodel2}. 
The energies of $|0\rangle_{h\gamma}^{\rm NS}$ and
$\alpha_{h\gamma;0}^\dagger|0\rangle_{h\gamma}^{\rm R}$ are given by 
\begin{equation}\label{ensr}
\begin{aligned}
\mathfrak{E}^{\rm NS}_{h\gamma}&=-\sum_{k\in {\rm NS}}\sqrt{(h-\cos k)^2+\gamma^2 \sin^2 k}\\
\mathfrak{E}^{\rm R}_{h\gamma}&=-\sum_{k\in {\rm R}}\sqrt{(h-\cos k)^2+\gamma^2 \sin^2 k}+2|1-h|\1_{h>1}\,.
\end{aligned}
\end{equation}
Here we have defined
\be
\1_{h>1}=\begin{cases}
1 & \text{if } h>1\ ,\\
0 & \text{else}\ .
\end{cases}
\ee

\subsection{Coherent states}
We define ${\rm NS}_+,{\rm R}_+$ as the subsets of NS and R
defined in \fr{NSR} with strictly positive elements. Given
$\pmb{k}\subset {\rm NS}_+$, we define pair states in the NS sector as
the Fock states
\begin{equation}
|\pmb{\bar{k}}\rangle_{h\gamma}=|\pmb{k}\cup(-\pmb{k})\rangle_{h\gamma}\,,
\end{equation}
and given $\pmb{k}\subset {\rm R}_+$, pair states in the R sector as
\begin{equation}
|\pmb{\bar{\bar{k}}}\rangle_{h\gamma}=|\pmb{k}\cup(-\pmb{k})\cup \{0\}\rangle_{h\gamma}\,.
\end{equation}
Following \cite{DBG}, for a complex number $A$ called ``phase" and a
function $f$ called ``amplitude", we introduce coherent states by
\begin{equation}\label{coherent}
\begin{aligned}
\Psi_{h\gamma}^{\rm NS}(A,f)&\equiv A\sum_{\pmb{k}\subset {\rm NS}_+} \left[\prod_{k\in\pmb{k}}f(k)\right]|\pmb{\bar{k}}\rangle=A\prod_{k\in {\rm NS}_+}\left[1+f(k)\alpha^\dagger_{h\gamma;-k}\alpha^\dagger_{h\gamma;k}\right]|0\rangle_{h\gamma}^{\rm NS}\\
\Psi_{h\gamma}^{\rm R}(A,f)&=A\sum_{\pmb{k}\subset {\rm R}_+} \left[\prod_{k\in\pmb{k}}f(k)\right]|\pmb{\bar{\bar{k}}}\rangle=A\prod_{k\in {\rm R}_+}\left[1+f(k)\alpha^\dagger_{h\gamma;-k}\alpha^\dagger_{h\gamma;k}\right]\alpha^\dagger_{h\gamma;0}|0\rangle_{h\gamma}^{\rm R}\,.
\end{aligned}
\end{equation}
In these definitions the amplitude $f$ needs to be defined only on
$[0,\pi]$. However we will consider it as an odd function defined on
$[-\pi,\pi]$. 

The key observation made in \cite{DBG} for the transverse field Ising chain
is the following relation between coherent states at different
at different parameter values $(h,\gamma)$ and $(\tilde{h},\tilde{\gamma})$:
\begin{theorem}
Let $h,\tilde{h}$ and $\gamma,\tilde{\gamma}$ be arbitrary magnetic
fields and anisotropies respectively. Then we have\label{thm} 
\begin{equation}
\Psi^{\rm NS,R}_{h\gamma}(A,f)=\Psi^{\rm NS,R}_{\tilde{h}\tilde{\gamma}}(\tilde{A},\tilde{f})\,,
\end{equation}
where
\begin{align}
\tilde{A}&=A \prod_{k\in{\rm NS}_+,{\rm
    R}_+}\frac{1+iK_{\tilde{h}\tilde{\gamma};h\gamma}(k)f(k)}{\sqrt{1+K_{\tilde{h}\tilde{\gamma};h\gamma}^2(k)}}\ ,\nn
\tilde{f}(k)&=\frac{iK_{\tilde{h}\tilde{\gamma};h\gamma}(k)+f(k)}{1+iK_{\tilde{h}\tilde{\gamma};h\gamma}(k)f(k)}\ .
\end{align}
Here we have defined
\begin{equation}
K_{\tilde{h}\tilde{\gamma};h\gamma}(k)=\tan
\frac{\theta_k^{\tilde{h}\tilde{\gamma}}-\theta_k^{h\gamma}}{2}\ ,\qquad
e^{i\theta_k^{h\gamma}}=\frac{h-\cos k-i\gamma\sin k}{\sqrt{(h-\cos k)^2+\gamma^2\sin^2 k}}\,.
\end{equation}

\end{theorem}
\begin{proof}
The proof is similar to that in \cite{DBG} for the Ising model. Expanding the coherent state in a basis of energy
  eigenstates gives
\begin{equation}
\Psi^{\rm NS}_{h\gamma}(A,f)=A\sum_{\pmb{q}\subset {\rm
    NS}}\sum_{\pmb{r}\subset {\rm NS}_+}\left[\prod_{r\in
    \pmb{r}}f(r)\right]|\pmb{q}\rangle_{\tilde{h}\tilde{\gamma}}^{\rm
  NS}\ \ {}^{\rm NS}_{\tilde{h}\tilde{\gamma}}\langle \pmb{q}|\pmb{\bar{r}}\rangle^{\rm NS}_{h\gamma}\,.
\end{equation}
The overlaps ${}^{\rm NS}_{\tilde{h}\tilde{\gamma}}\langle
\pmb{q}|\pmb{\bar{r}}\rangle^{\rm NS}_{h\gamma}$ between eigenstates
of $H(h,\gamma)$ at different magnetic fields and anisotropies is given in
Lemma \ref{overlap} in Appendix \ref{lemmas}. Introducing the
short-hand notation $K(k)\equiv K_{\tilde{h}\tilde{\gamma};h\gamma}(k)$ we have
\begin{align}
\Psi^{\rm NS}_{h\gamma}(A,f)&=\frac{A}{\displaystyle\prod_{k\in {\rm
      NS}_+}\sqrt{1+K^2(k)}}\sum_{\pmb{q}\subset {\rm
    NS}_+}\left[\prod_{q\in
    \pmb{q}}\left[iK(q)\right]\sum_{\pmb{r}\subset {\rm NS}_+}
  \prod_{r\in \pmb{r}}{\cal F}(r,\pmb{q})\right]
|\pmb{\bar{q}}\rangle_{\tilde{h}\tilde{\gamma}}^{\rm NS}\ ,\nn
{\cal F}(r,\pmb{q})&=\begin{cases}
\frac{f(r)}{iK(r)} & \text{ if }r\in \pmb{q}\ ,\\
iK(r)f(r) \qquad & \text{ if }r\notin \pmb{q}\ .
\end{cases}
\end{align}

The sum over $\pmb{r}$ is
\begin{equation}
\begin{aligned}
\sum_{\pmb{r}\subset {\rm NS}_+} \prod_{r\in \pmb{r}}
{\cal F}(r,\pmb{q})
&=\prod_{q\in \pmb{q}}\left(1+\frac{f(q)}{iK(q)}\right) \prod_{\substack{k\in {\rm NS}_+\\k\notin \pmb{q}}}(1+iK(k)f(k))\\
&=\prod_{q\in \pmb{q}}\frac{1+\frac{f(q)}{iK(q)}}{1+iK(q)f(q)} \prod_{\substack{k\in {\rm NS}_+}}(1+iK(k)f(k))\,.
\end{aligned}
\end{equation}
Then
\begin{equation}
\begin{aligned}
\Psi^{\rm NS}_{h\gamma}(A,f)=\tilde{A}\sum_{\pmb{q}\subset{\rm NS}_+} \left[\prod_{q\in \pmb{q}}\tilde{f}(q)\right]|\pmb{\bar{q}}\rangle_{\tilde{h}\tilde{\gamma}}^{\rm NS}\,,
\end{aligned}
\end{equation}
with $\tilde{A},\tilde{f}$ defined in the Theorem. 
\end{proof}

\subsection{Coherent Ensemble\label{sec:ce1}}
The purpose of this section is to introduce the Coherent Ensemble
which is convenient for formulating general time-dependent
  Hamiltonian dynamics.

\subsubsection{Generalized Gibbs and Gaussian ensembles}
We recall that the Generalized Gibbs Ensemble (GGE) parametrized by generalized temperatures $\beta_1,\beta_2,...$ is defined by the following expectation values of an operator $\mathcal{O}$
\begin{equation}
\langle \mathcal{O}\rangle^{\rm GGE[h\gamma]}_{\pmb{\beta}}=\frac{\tr[\mathcal{O} e^{-\sum_n \beta_n H_n}]}{\tr[e^{-\sum_n \beta_n H_n}]}\,,
\end{equation}
where $H_n$ are the conserved quantities of the model, and where $\tr$
denotes a trace over the full Hilbert space. These ensembles describe
equilibrium physics in the XY model, be it finite-temperature
equilibrium or steady states reached after a quantum quench. In the thermodynamic limit, they are equivalently parametrized by a particle density $\rho(\lambda)$ \cite{EF16}.

The Gaussian Ensemble (GE) parametrized by the $2\times 2$ block $L\times L$ correlation matrix $\Gamma$ is defined by the fact that the expectation values satisfy Wick's theorem when expressed in terms of the Jordan-Wigner fermions $c_j$, see Appendix \ref{xy}, the elementary $2$-point functions being given by
\begin{equation}
\Gamma_{ij}=\left(\begin{matrix}
\langle c^\dagger_ic_j\rangle^{\rm GE[h\gamma]}_{\Gamma}& \langle c^\dagger_ic^\dagger_j\rangle^{\rm GE[h\gamma]}_{\Gamma}\\
\langle c_ic_j\rangle^{\rm GE[h\gamma]}_{\Gamma}& \langle c_ic^\dagger_j\rangle^{\rm GE[h\gamma]}_{\Gamma}
\end{matrix} \right)\,.
\end{equation}
GGE's are particular cases of GE's for the XY Hamiltonian
\eqref{hamiltonian}. 

\subsubsection{Definition of the Coherent Ensemble}
An operator $\mathcal{O}$ is called even (resp. odd) if its matrix elements between eigenstates with different (resp. same) fermion parity vanish. We define the Coherent Ensemble (CE) parametrized by an amplitude
$f(k)$ by the following expectation values for even local operators
$\mathcal{O}$ 
\begin{equation}\label{ce}
\langle \mathcal{O}\rangle_f^{\rm CE[h\gamma]}=\Psi^{\rm NS}_{h\gamma}(A^{\rm NS},f)^\dagger \mathcal{O} \Psi^{\rm NS}_{h\gamma}(A^{\rm NS},f)\,,
\end{equation}
and for odd operators
\begin{equation}
\begin{aligned}\label{ceodd}
\langle \mathcal{O}\rangle_f^{\rm CE[h\gamma]}&=\frac{1}{2}\left[\Psi^{\rm R}_{h\gamma}(A^{\rm R},f)^\dagger \mathcal{O} \Psi^{\rm NS}_{h\gamma}(A^{\rm NS},f)+\Psi^{\rm NS}_{h\gamma}(A^{\rm NS},f)^\dagger \mathcal{O} \Psi^{\rm R}_{h\gamma}(A^{\rm R},f)\right]\\
&=\Re \left[\Psi^{\rm R}_{h\gamma}(A^{\rm R},f)^\dagger \mathcal{O} \Psi^{\rm NS}_{h\gamma}(A^{\rm NS},f)\right]\,,
\end{aligned}
\end{equation}
where
\be
A^{\rm NS,R}=\prod_{k\in{\rm NS}_+,{\rm
    R}_+}(1+|f(k)|^2)^{-\frac{1}{2}}\ .
\label{ARNS}
\ee
Replacing NS by R in \eqref{ce} incurs only negligible
corrections in system size. 
  
We note that the expression \eqref{ceodd} naturally arises when the
expectation value of an odd operator $\mathcal{O}$ is computed in a
state that is a superposition of NS and R sector states
\begin{equation}
|\Psi\rangle=\frac{|\Psi^{\rm NS}\rangle+|\Psi^{\rm R}\rangle}{\sqrt{2}}\,.
\end{equation}

\subsubsection{CE as a particular case of a GE}
Let us show that each CE corresponds to a particular GE. To that
end, we consider the following expectation value of Jordan-Wigner
fermions in momentum space 
\begin{equation}\label{obj}
\langle c(k_1)^\dagger...c(k_n)^\dagger c(q_1)...c(q_m)\rangle_f^{\rm CE[h\gamma]}\,,
\end{equation}
with for example $k_1,...,k_n,q_1,...,q_m\in{\rm NS}$, and would like to show that it can be computed using Wick contractions. Using Theorem \ref{thm}, we can write it in the $(\infty,\gamma)$ basis with another amplitude $f'$. In this basis we have
\begin{equation}
\begin{aligned}
&\langle c(k_1)^\dagger...c(k_n)^\dagger c(q_1)...c(q_m)\rangle_f^{\rm CE[h\gamma]}=\langle \alpha_{k_1}^\dagger...\alpha_{k_n}^\dagger \alpha_{q_1}...\alpha_{q_m}\rangle_{f'}^{\rm CE[\infty\gamma]}\,,
\end{aligned}
\end{equation}
where the $\alpha$'s are implicitly written in the $(\infty,\gamma)$
basis for notational lightness. Next, we observe that for the
expectation value $\langle \alpha_{k_1}^\dagger...\alpha_{k_m}^\dagger
\alpha_{q_1}...\alpha_{q_m}\rangle_{f'}^{\rm CE[\infty\gamma]}$ to be
non-zero, for each $k_i$ there has to be either another $k_j$ with
$k_j=-k_i$, or a $q_j$ with $q_j=k_i$. The same holds true
interchanging the $k$'s and the $q$'s. Hence we are led to
evaluating expectation values of the form 
\begin{equation}
\langle \prod_{i}\alpha_{k_i}^\dagger \alpha_{k_i}\prod_{i}\alpha_{-q_i}^\dagger \alpha_{q_i}^\dagger\prod_i\alpha_{r_i} \alpha_{-r_i}\prod_i \alpha_{-s_i}^\dagger \alpha_{s_i}^\dagger\alpha_{s_i} \alpha_{-s_i} \rangle_{f'}^{\rm CE[\infty\gamma]}\,,
\end{equation}
where the $k$'s, $q$'s, $r$'s and $s$'s are all distinct. Using the
definition of the CE, we obtain
\begin{equation}\label{wick}
\begin{aligned}
&\langle \prod_{i}\alpha_{k_i}^\dagger \alpha_{k_i}\prod_{i}\alpha_{-q_i}^\dagger \alpha_{q_i}^\dagger\prod_i\alpha_{r_i} \alpha_{-r_i}\prod_i \alpha_{-s_i}^\dagger \alpha_{s_i}^\dagger\alpha_{s_i} \alpha_{-s_i} \rangle_{f'}^{\rm CE[\infty\gamma]}\\
&=\prod_i \frac{|f'(k_i)|^2}{1+|f'(k_i)|^2}\prod_i \frac{f^{\prime *}(q_i)}{1+|f'(q_i)|^2}\prod_i \frac{f'(r_i)}{1+|f'(r_i)|^2}\prod_i \frac{|f'(s_i)|^2}{1+|f'(s_i)|^2}\,.
\end{aligned}
\end{equation}
We now observe that the non-zero elementary two-point functions satisfy
\begin{equation}
\langle\alpha_{k}^\dagger \alpha_{k} \rangle_{f'}^{\rm CE[\infty\gamma]}=\frac{|f'(k)|^2}{1+|f'(k)|^2}\,,\qquad \langle\alpha_{k}\alpha_{-k} \rangle_{f'}^{\rm CE[\infty\gamma]}=\frac{f'(k)}{1+|f'(k)|^2}\,.
\end{equation}
Because of the relation
\begin{equation}
\frac{|f'(s)|^2}{1+|f'(s)|^2}=\langle\alpha_{-s}^\dagger \alpha_{s}^\dagger \rangle_{f'}^{\rm CE[\infty\gamma]}\langle\alpha_{s}\alpha_{-s} \rangle_{f'}^{\rm CE[\infty\gamma]}+\langle\alpha_{-s}^\dagger \alpha_{-s} \rangle_{f'}^{\rm CE[\infty\gamma]}\langle\alpha_{s}^\dagger\alpha_{s} \rangle_{f'}^{\rm CE[\infty\gamma]}\,,
\end{equation}
we obtain that the right-hand side of \eqref{wick} and so \eqref{obj}
can indeed be computed using Wick's theorem, which
  establishes that the CE is a particular case of a GE.

\subsubsection{Inequivalence of CE with GE or GGE}

CE ensembles are \textit{not} equivalent to either GEs or
GGEs. To show this, let us consider an operator $\mathcal{O}$ that
is local in terms of the fermions $c_j$ and compute its expectation
value within the CE. Using Wick's theorem in the thermodynamic limit
it can be recast into sums and products of expectation values of
quadratic terms in the Jordan Wigner fermions $c_j$'s in real
space. These take the values 
\begin{align}
\label{CMce}
\langle c_n^\dagger c_m\rangle^{\rm
  CE[h\gamma]}_{f}&=\frac{1}{2\pi}\int_{-\pi}^\pi e^{i(n-m)k}
\left[\cos^2(\theta_{k}^{h\gamma}/2)\frac{|f(k)|^2}{1+|f(k)|^2}
+\sin^2(\theta_{k}^{h\gamma}/2)\frac{1}{1+|f(k)|^2}\right]
\D{k}\nn
&\qquad-\frac{1}{2\pi}\int_{-\pi}^\pi e^{i(n-m)k}\sin
\theta_{k}^{h\gamma} \frac{\Im f(k)}{1+|f(k)|^2}\D{k}\ ,\nn 
\langle c_n c_m\rangle^{\rm CE[h\gamma]}_{f}&=\frac{i}{4\pi}\int_{-\pi}^\pi e^{ik(n-m)}\sin\theta_{k}^{h\gamma} \frac{1-|f(k)|^2}{1+|f(k)|^2}\D{k}\nn
&-\frac{1}{2\pi}\int_{-\pi}^\pi
e^{ik(n-m)}\frac{\cos^2(\theta_{k}^{h\gamma}/2)
  f(k)+\sin^2(\theta_{k}^{h\gamma}/2) f^*(k)}{1+|f(k)|^2}\D{k}\ .
\end{align}
By computing the Fourier series, one sees that the values of $\langle
c_n c_m\rangle^{\rm CE[h\gamma]}_{f}$ for all $n,m$ impose a system of
two polynomial equations of degree $2$ on $\Re f$ and $\Im f$. This
prevents $\langle c_n^\dagger c_m\rangle^{\rm CE[h\gamma]}_{f}$ from
taking arbitrary values, whereas in the GE they are independent
quantities. Hence the CE's are a strict subset of GE's.

Expectation values of local operators in the thermodynamic limit of a
GGE can be expressed in terms of mode occupation numbers or
equivalently a root density $\rho$ \cite{EF16} and the associated hole density
$\rho_h=\tfrac{1}{2\pi}-\rho$ as
\begin{equation}\label{ee}
\begin{aligned}
\langle c_n^\dagger c_m\rangle^{\rm GGE[h\gamma]}_{\rho}&=\int_{-\pi}^\pi e^{i(n-m)k}\cos^2(\theta_{k}^{h,\gamma}/2)\rho(k)\D{k}+\int_{-\pi}^\pi e^{i(n-m)k}\sin^2(\theta_{k}^{h\gamma}/2)\rho_h(-k)\D{k}\\
\langle c_n c_m\rangle^{\rm GGE[h\gamma]}_{\rho}&=\frac{i}{2}\int_{-\pi}^\pi e^{ik(n-m)}\sin\theta_{k}^{h,\gamma} (\rho_h(-k)-\rho(k))\D{k} \,.
\end{aligned}
\end{equation}
To have $\langle c_n c_m\rangle^{\rm GGE[h\gamma]}_{\rho}=\langle c_n
c_m\rangle^{\rm CE[h\gamma]}_{f}$ for all $n,m$ requires a
purely imaginary $f(k)\equiv i\tilde{f}(k)$ and the relation 
\begin{equation}
\rho(k)=\frac{1}{2\pi}\frac{\tilde{f}(k)^2}{1+\tilde{f}(k)^2}+\frac{1}{2\pi\tan \theta^{h\gamma}_k} \frac{\tilde{f}(k)}{1+\tilde{f}(k)^2}\,.
\end{equation}
The requirement that $\langle c^\dagger_n c_m\rangle^{\rm
  GGE[h\gamma]}_{\rho}=\langle c^\dagger_n c_m\rangle^{\rm
  CE[h\gamma]}_{f}$ for all $n,m$ further imposes that
\begin{equation}
\rho(k)=\frac{1}{2\pi}\frac{\tilde{f}(k)^2}{1+\tilde{f}(k)^2}-\frac{\tan \theta^{h\gamma}_k}{2\pi} \frac{\tilde{f}(k)}{1+\tilde{f}(k)^2}\,.
\end{equation}
One sees that the two relations are compatible only if $\tilde{f}(k)$ takes
the values $0$ or $\infty$. In this case, the coherent state
$\Psi_{h\gamma}^{\rm NS}(A,f)$ is nothing but an eigenstate of the
Hamiltonian. In fact, it is a ``representative state''
\cite{EF16} of a root density that is  either zero or maximal, which
exactly corresponds to so-called ``zero-entropy states", in the sense
that their Yang-Yang entropy vanishes. Hence no GGE can be
written as a CE, apart from zero-entropy state expectation values. 

\subsubsection{GGE at the boundary of CE}
However, starting from a coherent state $\Psi_{h\gamma}^{\rm NS}(f)$
one can obtain a GGE by taking the late time limit of the
evolution of the CE induced by the Hamiltonian
$H(h,\gamma)$. Indeed, one has  
\begin{equation}\label{evoty}
e^{-i\tau H(h,\gamma)}\Psi_{h\gamma}^{\rm NS,R}(f,A)=\Psi_{h\gamma}^{\rm NS,R}(f_\tau,A_\tau)\,,
\end{equation}
with
\begin{equation}
f_\tau(k)=f(k)e^{-2i\tau\varepsilon_{h\gamma}(k)}\,,\qquad A_\tau=A e^{-i\tau \mathfrak{E}^{\rm NS,R}_{h\gamma}}.
\end{equation}
Hence the CE after time $\tau$ is obtained from  \eqref{CMce} by
replacing $f$ by $f_\tau$. In the limit $\tau\to\infty$, the fast
oscillations in $f_\tau(k)$ cause the second terms of the
expectation values in \eqref{CMce} to vanish, while leaving the other
terms invariant. This establishes that if the root density $\rho$ is
even, and if one chooses $f$ such that 
\begin{equation}\label{rho}
\rho(k)=\frac{1}{2\pi}\frac{|f(k)|^2}{1+|f(k)|^2}\,,
\end{equation}
then
\begin{equation}
\underset{\tau\to\infty}{\lim}\,\langle c_n^\dagger c_m\rangle^{\rm CE[h\gamma]}_{f_\tau}=\langle c_n^\dagger c_m\rangle^{\rm GGE[h\gamma]}_{\rho}\,,\qquad \underset{\tau\to\infty}{\lim}\,\langle c_n c_m\rangle^{\rm CE[h\gamma]}_{f_\tau}=\langle c_n c_m\rangle^{\rm GGE[h\gamma]}_{\rho}\,.
\end{equation}
This shows that any GGE with even root density can be obtained as a limit of CE. Figure \ref{fig:ce} summarizes the inclusion of the different ensembles GE, CE and GGE.

\begin{figure}[H]
\begin{center}
 \includegraphics[scale=0.45]{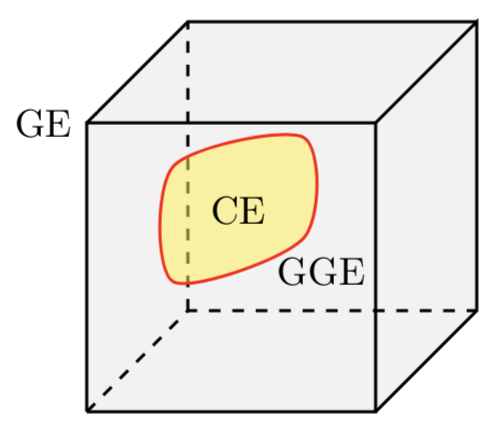}
\end{center}
\caption{Sketch of the position of the GGE (red line) and CE (yellow surface) within the GE (gray volume), for symmetric root densities.}
\label{fig:ce}
\end{figure}

\subsection{Out-of-equilibrium physics as a Coherent Ensemble 
\label{outofeq}}

It is known that equilibrium physics can be formulated as a GGE
\cite{EF16}. The purpose of this section is to show how homogeneous
non-equilibrium dynamics in the XY model can be formulated as a
CE, see also Refs\cite{Dziarmaga05,Cherng06,Mukherjee07,Dziarmaga10,Schuricht16}.


\subsubsection{Differential equation for the amplitude} 
We assume that at time $t=0$ the system is prepared in the ground state of $H(h_0,\gamma_0)$, and that it is time-evolved at time $t>0$ with the Hamiltonian $H(h(t),\gamma(t))$, namely
\begin{align}\label{eqdiff}
|\psi(0)\rangle&=|0\rangle_{h_0\gamma_0}^{\rm NS}\ ,\nn
i\partial_t |\psi(t)\rangle&=H(h(t),\gamma(t))|\psi(t)\rangle\ .
\end{align}
We note that by virtue of the linearity of the Schr\"odinger
equation one can equally well consider initial states that are
superpositions, for example of the ground states in the NS and R sectors. 

We would like to determine the time evolution of observables during
this process.  To that end, we replace the time evolution of the
magnetic field and anisotropy by a series of quenches in which they
are suddenly changed to $h_n=h(t_n),\gamma_n=\gamma(t_n)$ at times
$t_n=(n-1) \delta t$, for a given small time interval $\delta t>0$,
and kept constant between these quenches. The original dynamics is
obtained in the limit $\delta t\to 0$. We now observe that the initial
state can be written as a coherent state 
\begin{equation}
|\psi(0)\rangle=\Psi^{\rm NS}_{h_0\gamma_0}(1,0)\,,
\label{initialstate}
\end{equation}
and that the time-evolution with $H(h,\gamma)$ of a coherent state written in the $(h,\gamma)$ basis is simply given by
\begin{equation}
e^{-itH(h,\gamma)}\Psi^{\rm NS}_{h\gamma}(A,f)=\Psi^{\rm
  NS}_{h\gamma}(A',f')\ ,
\end{equation}
with $f'(k)=f(k)e^{-2it\varepsilon_{h\gamma}(k)}$ and
$A'=Ae^{-it\mathfrak{E}^{\rm NS}_{h\gamma}}$. As a consequence, using
repeatedly Theorem \ref{thm} to write the state as a coherent state in
the $(h_n,\gamma_n)$ basis for $t_n\leq t<t_{n+1}$, and then
expressing it in the  $(0,1)$ basis, one has at time $t_n^-$ 
\begin{equation}
|\psi(t_n^-)\rangle=\Psi^{\rm NS}_{01}(A^{\rm NS}_{(n-1)},f_{(n-1)})\ ,
\end{equation}
where the sequence of functions $f_j$ and phases $A_j$ satisfy
\begin{align}
\label{recf}
f_{(0)}(k)&=-iK_{h_0\gamma_0;01}(k)\ ,\nn
f_{(j)}(k)&=\frac{iK_{h_j\gamma_j;01}(k)(e^{-2i
    \varepsilon_{h_{j}\gamma_{j}}(k)\delta
    t}-1)+(K^2_{h_j\gamma_j;01}(k)+e^{-2i
    \varepsilon_{h_{j}\gamma_{j}}(k)\delta t})f_{(j-1)}(k)}{1+e^{-2i
    \varepsilon_{h_{j}\gamma_{j}}(k)\delta
    t}K^2_{h_j\gamma_j;01}(k)+iK_{h_j\gamma_j;01}(k)(1-e^{-2i
    \varepsilon_{h_{j}\gamma_{j}}(k)\delta t})f_{(j-1)}(k)}\ , \nn
A^{\rm NS}_{j}&=A^{\rm NS}_{j-1}e^{-i\delta t \mathfrak{E}^{\rm
    NS}_{h_j\gamma_j}}\ ,\nn
&\times\prod_{k\in{\rm NS}_+}\frac{1+K^2_{h_j\gamma_j;01}(k)e^{-2i
    \varepsilon_{h_{j}\gamma_{j}}(k)\delta t}+i
  K_{h_j\gamma_j;01}(k)f_{(j-1)}(k)(1-e^{-2i
    \varepsilon_{h_{j}\gamma_{j}}(k)\delta
    t})}{1+K^2_{h_j\gamma_j;01}(k)}\ .
\end{align}
We now take the limit $\delta t\to 0$. To that end it is useful to
introduce a function $f_t(k)$ of both $t$ and $k$ by
\begin{equation}
f_t(k)=\underset{\delta t\to 0}{\lim}\, f_{(\lfloor t/\delta t\rfloor)}(k)\,.
\end{equation}
From \eqref{recf}, we conclude that the state of the system at time $t$ following an arbitrary variation $h(t),\gamma(t)$ of the magnetic field and anisotropy can be written as a coherent state
\begin{equation}
|\psi(t)\rangle=\Psi^{\rm NS}_{01}(A^{\rm NS}_{t},f_{t})\,,
\label{timeevolved}
\end{equation}
whose amplitude $f_t(k)$ satisfies a non-linear differential equation
\begin{equation}\label{equadiff}
\partial_tf_t(k)=\frac{2K_{h(t)\gamma(t);01}(k)}{1+K_{h(t)\gamma(t);01}^2(k)}\varepsilon_{h(t)\gamma(t)}(k)(1+f^2_t(k))-2i\frac{1-K_{h(t)\gamma(t);01}^2(k)}{1+K_{h(t)\gamma(t);01}^2(k)}\varepsilon_{h(t)\gamma(t)}(k)f_t(k)\ .
\end{equation}
The initial condition is $f_0(k)=-iK_{h(0)\gamma(0);01}(k)$. This
shows that any expectation value out-of-equilibrium can be written as
a CE. An example of the function $f_t(k)$ is plotted in Figure
\ref{ctdensity5} for a sudden quench from $h_0=0.1$ to $h=0.9$.

An equivalent system of linear differential equations was obtained
previously in \cite{Dziarmaga05}. Indeed, we have 
\begin{align}
|\psi(t)\rangle&=\prod_{k\in{\rm NS}_+}\left[n_t(k)+m_t(k)\alpha^\dagger_{01;-k}
  \alpha^\dagger_{01;k}\right]|0\rangle_{01}^{\rm  NS}\ ,
\label{psioft}
\end{align}
where $n_t(p)$ an $m_t(p)$ fulfil the following system of linear ordinary
differential equations 
\begin{align}
\frac{d}{dt}\begin{pmatrix}
  n_{t}(p)\\
  m_t(p)
\end{pmatrix}=
\varepsilon_{h(t),\gamma(t)}(p)
\begin{pmatrix}
i\cos\Delta_t(p) & -\sin\Delta_t(p)\\
\sin\Delta_t(p) &-i\cos\Delta_t(p)
\end{pmatrix}
\begin{pmatrix}
n_{t}(p)\\
m_t(p)
\end{pmatrix}\ ,\quad p\in{\rm NS_+}\,,
\label{system0}
\end{align}
where we have defined 
\begin{align}
\Delta_t(k)&=\theta^{h(t)\gamma(t)}_k-\theta^{01}_k\,,
\end{align}
and with the initial conditions
\begin{align}
n_0(k)&=\frac{1}{\sqrt{1+K^2_{h_0\gamma_0;01}(k)}}\ ,\quad
m_0(k)=-\frac{iK_{h_0\gamma_0;01}(k)}{\sqrt{1+K^2_{h_0\gamma_0;01}(k)}}\ .
\end{align}
This formulation is equivalent to \fr{timeevolved} once we identify
\begin{align}
f_t(p)&=\frac{m_t(p)}{n_t(p)}\ ,\quad A^{\rm NS}_t=\prod_{p\in{\rm
    NS}_+}n_t(p)\ .
\end{align}

\begin{figure}[ht!]
\begin{center}
 \includegraphics[scale=0.425]{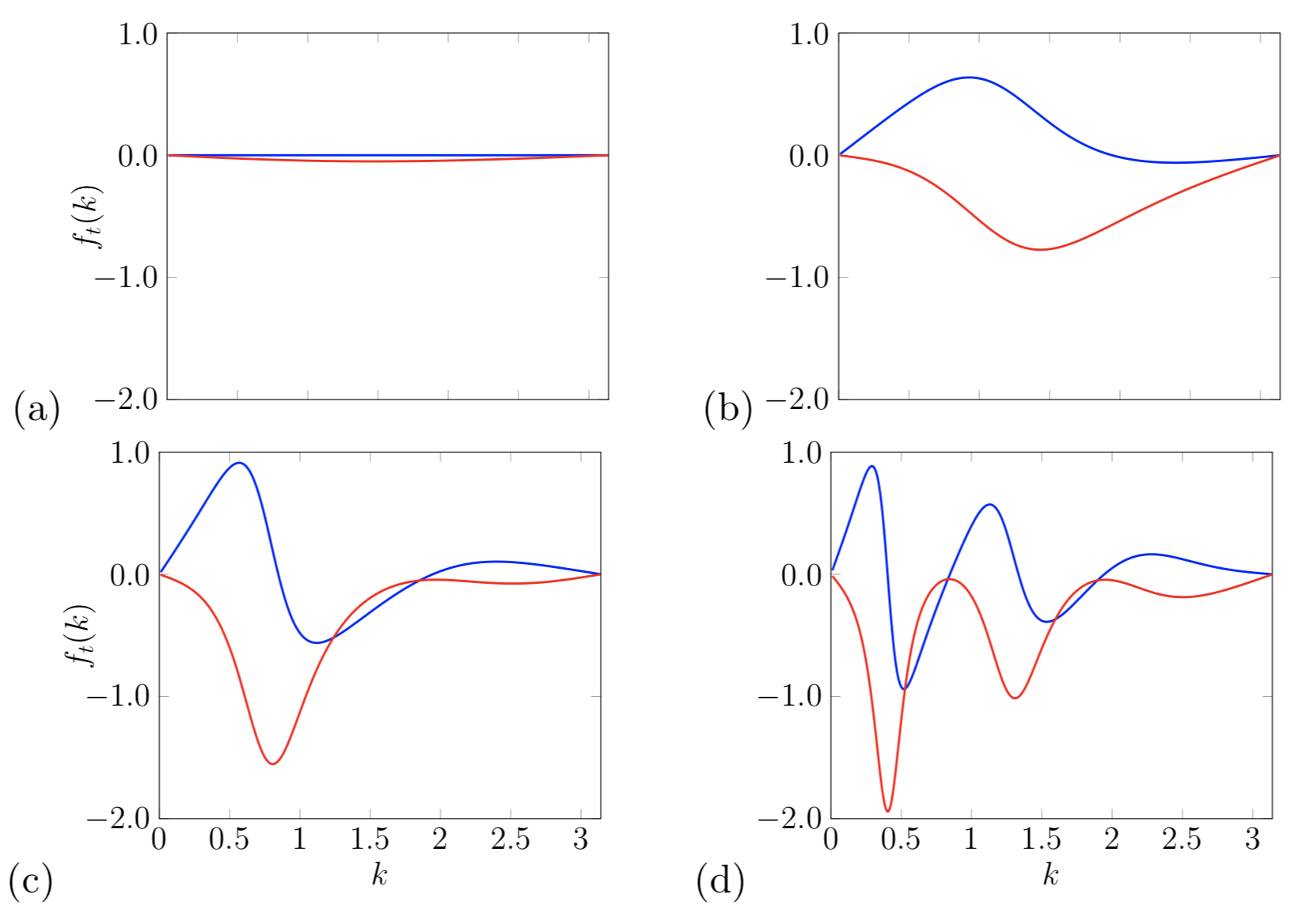}
\caption{Amplitude $f_t(k)$ for a quantum quench from $h(0)=0.1$
    to $h(t)=0.9$ for $t>0$ at times (a) $t=0$; (b) $t=0.5$; (c)
  $t=1$; (d) $t=2$. The real and imaginary parts are shown in blue and
  red respectively.}
\label {ctdensity5}
\end{center}
\end {figure}

\subsubsection{The phase factor \label{secpha}}
In the limit $\delta t\to 0$, the phase becomes
\begin{equation}\label{dt0phase}
\begin{aligned}
&A_t^{\rm NS}=A_0^{\rm NS}e^{-i \int_0^t \mathfrak{E}^{\rm NS}_{h(s)\gamma(s)}\D{s}}\exp\left(\sum_{k\in{\rm NS}_+}\varphi_t(k) \right)\\
&\varphi_t(k)=-2\int_0^t \varepsilon_{h(s)\gamma(s)}(k) \frac{K_{h(s)\gamma(s);01}(k)}{1+K^2_{h(s)\gamma(s);01}(k)}(iK_{h(s)\gamma(s);01}(k)+f_s(k))\D{s} \,.
\end{aligned}
\end{equation}
For a coherent state in the R sector the same formula holds where the
sum is over momenta in ${\rm R}_+$ and with $\mathfrak{E}^{\rm
  NS}_{h(s)\gamma(s)}$ replaced by $\mathfrak{E}^{\rm 
  R}_{h(s)\gamma(s)}$. For expectation values in the CE of even
operators, the phase is irrelevant since it always cancels
out. However, for odd operators the expectation value is proportional
to the phase factor
\begin{equation}
\phi_L(t)\equiv \frac{A^{\rm NS}_t (A_t^{\rm R})^*}{|A^{\rm NS}_t (A_t^{\rm R})^*|}\,,
\label{phiLoft}
\end{equation}
where we made explicit the system size dependence of $\phi_L(t)$ that
is only implicit in $A^{\rm NS,R}_t$.

Let us assume first that the solution $f_t(k)$ to the non-linear
differential equation \eqref{equadiff} is regular for all $t$ and
$k$. Then, using the Euler-MacLaurin formula, we find 
\begin{equation}
\begin{aligned}
\sum_{k\in{\rm NS}_+}\varphi_t(k)&=\frac{L}{2\pi} \int_0^{\pi}\varphi_t(k)\D{k}+\mathcal{O}(L^{-1})\\
\sum_{k\in{\rm R}_+}\varphi_t(k)&=\frac{L}{2\pi} \int_0^{\pi}\varphi_t(k)\D{k}-\frac{\varphi_t(\pi)+\varphi_t(0)}{2}+\mathcal{O}(L^{-1})\\
\mathfrak{E}^{\rm NS}_{h\gamma}&=-\frac{L}{4\pi}\int_{-\pi}^\pi \varepsilon_{h\gamma}(k)\D{k}+\mathcal{O}(L^{-1})\\
\mathfrak{E}^{\rm R}_{h\gamma}&=-\frac{L}{4\pi}\int_{-\pi}^\pi \varepsilon_{h\gamma}(k)\D{k}+2 |1-h|\1_{h>1}+\mathcal{O}(L^{-1})\,.
\end{aligned}
\end{equation}
Assuming that the trajectory $h(t),\gamma(t)$ is such that the time spent on a critical point is of measure $0$, we find $\varphi_t(\pi)=0$ and
\begin{equation}
\varphi_t(0)=-4i \int_0^t |1-h(s)|\1_{h(s)>1}\D{s}\,.
\end{equation}
Hence in this case we obtain in the thermodynamic limit
\begin{equation}
\phi_\infty(t)=1\,.
\end{equation}
However, if the function $f_t(k)$ is singular for some values
$t^*,k^*$, then $\phi_L(t)$ for $t\geq t^*$ is not guaranteed to
become $1$ in the thermodynamic limit, which can result in a 
non-trivial multiplicative phase in \eqref{ceodd}. This phase has to be computed with \eqref{phiLoft} and \eqref{dt0phase}.

Singularities of $f_t(k)$ are best understood with the system of linear differential equations \eqref{system0}. $n_t(p)$ and $m_t(p)$ are regular functions and the nature of the
singularities of $f_t(p)$ becomes transparent: they simply correspond
to situations when at least one probability amplitude $n_{t^*}(k^*)$
vanishes. This implies that the overlap of the time evolved state with
$|0\rangle^{\rm NS}_{01}$ vanishes
\be
{}_{01}^{\rm NS}\langle 0|T\exp\left(-i\int_0^{t^*} H(h(s),\gamma(s))\D{s}\right)|0\rangle^{\rm
  NS}_{h_0\gamma_0}=0 \implies f_{t^*}(k^*)\text{ singular}.
\ee
This situation is somewhat reminiscent of non-analyticities in the
Loschmidt amplitude \cite{DPT}. This phase will be discussed again in
a concrete example in Section \ref{KZsec}. \\

To summarize this section, the CE provides the natural framework to
evaluate the expectation value of any operator $\mathcal{O}$ during
the out-of-equilibrium evolution \eqref{eqdiff}. If $\mathcal{O}$ is
even, then its expectation value is given by \eqref{ce} with $f$
satisfying the nonlinear differential equation \eqref{equadiff}. If
$\mathcal{O}$ is odd, then its expectation value is given by
\eqref{ceodd} multiplied (inside the real part) by the phase
$\phi_L(t)$ \eqref{phiLoft}. In the thermodynamic limit,
$\phi_\infty(t)$ is constant in time as long as $f_t(k)$ is a regular
function of $k$. If $f_{t^*}(k)$ has a singularity at $k^*$, then
$\phi_\infty(t)$ can be discontinuous at $t^*$, and has to be
evaluated according to \eqref{dt0phase}.

\section{Expectation values in the Coherent Ensemble \label{sec:ce}}
The purpose of this section is to show that essentially all
correlation functions and expectation values in the CE can be 
expressed as Fredholm determinants and Pfaffians in the thermodynamic limit.
We fix $h,\gamma$ and to ease notations write
\begin{equation}
\langle \mathcal{O}\rangle_f\equiv\langle \mathcal{O}\rangle^{\rm CE[h\gamma]}_f\,.
\end{equation}
\subsection{Definitions}
The formulas obtained for the various expectation values
considered involve Fredholm determinants and Fredholm Pfaffians. In
this section we present the definition of these objects and some of
their properties.  
\subsubsection{Fredholm determinant}
Given a function $F(\lambda,\mu)$ on $[a,b]\times [a,b]$, the Fredholm determinant $\Det[{\rm Id}+F]$ is defined by
\begin{equation}
\Det[{\rm Id}+F]=1+\sum_{n= 1}^\infty \frac{1}{n!}\int_{a}^{b}...\int_{a}^b \det[F(z_i,z_j)]_{1\leq i,j\leq n}\D{z_1}...\D{z_n}\,.
\end{equation}
Here, ${\rm Id}$ should be merely considered as a notation. The Fredholm determinant satisfies the following relation
\begin{equation}\label{fredheq}
\Det[{\rm Id}+F]=\underset{N\to\infty}{\lim}\det\left[\delta_{i,j}+\frac{b-a}{N}F(\zeta_i,\zeta_j)\right]_{1\leq i,j\leq N}\,,
\end{equation}
with $\zeta_1<...<\zeta_N$ regularly spaced numbers covering $[a,b]$.

\subsubsection{Block Fredholm Pfaffian}
Given a $2\times 2$ matrix-valued function $\pmb{K}(x,y)=(K_{ij}(x,y))_{1\leq i,j\leq 2}$ on $[a,b]\times[a,b]$ satisfying $K_{ij}(x,y)=-K_{ji}(y,x)$, the block Fredholm Pfaffian $\Pf[\pmb{{\rm Jd}}+\pmb{K}]$ is defined by \cite{rains}
\begin{equation}\label{pfa}
\Pf[\pmb{{\rm Jd}}+\pmb{K}]=1+\sum_{n= 1}^\infty \frac{1}{n!}\int_{a}^{b}...\int_{a}^b \pf[K(z_i,z_j)]_{1\leq i,j\leq n}\D{z_1}...\D{z_{n}}\,.
\end{equation}
The matrices inside the Pfaffian on the right-hand side are thus $n\times n$ matrices of $2\times 2$ blocks. Here, $\pmb{{\rm Jd}}$ should be merely considered as a notation. The block Fredholm Pfaffian satisfies the relation
\begin{equation}\label{111}
\Pf[\pmb{{\rm Jd}}+\pmb{K}]=\underset{\substack{N\to\infty}}{\lim} \pf\left[\delta_{i,j}\pmb{J}+\frac{b-a}{N}\pmb{K}(\zeta_i,\zeta_j)\right]_{1\leq i,j\leq N}\,,
\end{equation}
with $J$ the $2\times 2$ matrix
\begin{equation}
\pmb{J}=\left(\begin{matrix}0&1\\-1&0\end{matrix} \right)\,.
\end{equation}
\subsubsection{Fredholm Pfaffian}
Given an antisymmetric function $F(\lambda,\mu)$ on $[-a,a]\times[-a,a]$, i.e. that satisfies $F(\mu,\lambda)=-F(\lambda,\mu)$, one can define the $2\times 2$ matrix-valued function $\pmb{K}_F$ on $[0,a]\times[0,a]$ by
\begin{equation}
\pmb{K}_F(x,y)=\left(\begin{matrix}F(x,y)&F(x,-y)\\F(-x,y)&F(-x,-y)\end{matrix} \right)\,.
\end{equation}
We thus define the Fredholm Pfaffian $\Pf[{\rm Jd}+F]$ of an antisymmetric function $F$ on $[-a,a]\times[-a,a]$ by
\begin{equation}\label{pfafdef}
\begin{aligned}
\Pf[{\rm Jd}+F]&\equiv\Pf[\pmb{{\rm Jd}}+\pmb{K}_F]\\
&=1+\sum_{n= 1}^\infty \frac{1}{n!}\int_{0}^{a}...\int_{0}^a \pf\left[\begin{matrix}F(z_i,z_j)&F(z_i,-z_j)\\ F(-z_i,z_j)&F(-z_i,-z_j)\end{matrix}\right]_{1\leq i,j\leq n}\D{z_1}...\D{z_{n}}\,.
\end{aligned}
\end{equation}
It satisfies the relation
\begin{equation}
\Pf[{\rm Jd}+F]=\underset{\substack{N\to\infty\\ N {\rm
      even}}}{\lim}(-1)^{N/2}
\pf\left[\delta_{i,N+1-j}\sign(j-i)+\frac{2a}{N}F(\zeta_i,\zeta_j)\right]_{1\leq
  i,j\leq N}\ .
\end{equation}
Here, $\zeta_1<...<\zeta_N$ are regularly spaced numbers covering $[-a,a]$ and assumed to be symmetrically distributed to ensure the antisymmetry of the matrix. The
factor $(-1)^{N/2}$ compared to \eqref{111} arises from the
re-ordering of rows and columns after changing the $2\times 2$ block
$N/2\times N/2$ matrix into an $N\times N$ matrix, and re-ordering the negative $\zeta$'s in ascending order. 


\subsection{Full counting statistics of the transverse magnetization}
As the operator $\sigma^z$ is local in the Jordan-Wigner fermions
$c_j$, any static correlation of $\sigma^z$ is simple to calculate and
can be expressed as a multiple integral in the thermodynamic limit. The
purpose of this section is to derive a Fredholm determinant expression
for the following generating function 
\begin{equation}\label{gege}
\langle e^{i\theta \sum_{j=1}^{\ell} \sigma^z_j}\rangle_f\,,
\end{equation}
for arbitrary $\theta$ and $\ell$. Exact Pfaffian representations
of size $2\ell$ for the full counting statistics of the transverse
magnetization in a generic GE have been derived before using Wick's
theorem \cite{Cherng06,cd-07,e-13,ia-13,klich14,CoEG17,Groha18}. 

To compute \eqref{gege}, we express the coherent states involved in the CE in the
$(\infty,\gamma)$ basis and expand them to obtain 
\begin{equation}
\langle e^{i\theta \sum_{j=1}^{\ell} \sigma^z_j}\rangle_f=|A|^2\sum_{\pmb{\lambda},\pmb{\mu}\subset{\rm NS}_+}{}^{\rm NS}_{\infty\gamma}\langle \pmb{\bar{\lambda}}|e^{i\theta \sum_{j=1}^{\ell} \sigma^z_j}|\pmb{\bar{\mu}}\rangle_{\infty\gamma}^{\rm NS}\prod_{\lambda\in\pmb{\lambda}}g^*(\lambda)\prod_{\mu\in\pmb{\mu}}g(\mu)\,,
\end{equation}
where
\begin{equation}\label{ggg}
g(k)=\frac{iK_{\infty\gamma;h\gamma}(k)+f(k)}{1+iK_{\infty\gamma;h\gamma}(k)f(k)}\,.
\end{equation}
We now use Lemma \ref{ff1} to write the form factor of $e^{i\theta
  \sum_{j=1}^{\ell}\sigma^z_j}$ as a determinant 
\begin{equation}
\langle e^{i\theta \sum_{j=1}^{\ell} \sigma^z_j}\rangle_f=e^{i\theta\ell}|A|^2\sum_{\substack{\pmb{\lambda},\pmb{\mu}\subset{\rm NS}_+\\ |\pmb{\lambda}|=|\pmb{\mu}|}}\det E(\pmb{\bar{\lambda}},\pmb{\bar{\mu}})\prod_{\lambda\in\pmb{\lambda}}g^*(\lambda)\prod_{\mu\in\pmb{\mu}}g(\mu)\,,
\label{FCSeq1}
\end{equation}
where $E(\pmb{\lambda},\pmb{\mu})$ is defined in
\eqref{E}. Because of the pair structure of $\pmb{\bar{\mu}}$ each
$\mu_i$ appears in two columns of the matrix
$E(\pmb{\bar{\lambda}},\pmb{\bar{\mu}})$. Hence we can use Lemma
\ref{brujin} to carry out the sum over $\pmb{\mu}$, with ${\rm NS}_+$ being the set $K$, and \eqref{E} being the function $f$ when $\mu_k>0$ and the function $g$ when $\mu_k<0$. This gives
\begin{equation}\label{eq1}
\sum_{\substack{\pmb{\mu}\subset{\rm
      NS}_+\\ |\pmb{\lambda}|=|\pmb{\mu}|}}\det
E(\pmb{\bar{\lambda}},\pmb{\bar{\mu}})\prod_{\mu\in\pmb{\mu}}g(\mu)=(-1)^{N/2}\pf
[\tilde{E}(\pmb{\bar{\lambda}})-\tilde{E}(\pmb{\bar{\lambda}})^T]\ .
\end{equation}
In notations where $\pmb{\bar{\lambda}}=\{\lambda_1,...,\lambda_N\}$
the matrix elements of $\tilde{E}$ for $\lambda_j\neq -\lambda_k$ are
given by
\begin{equation}
\begin{aligned}
&\tilde{E}(\pmb{\bar{\lambda}})_{jk}=\left(\frac{e^{-2i\theta}-1}{L}\right)^2\sum_{\substack{\mu\in{\rm NS}_+\\\mu\neq \lambda_j,-\lambda_k}} e^{i(\lambda_j+\lambda_k)}\frac{1-e^{i\ell(\lambda_j-\mu)}}{1-e^{i(\lambda_j-\mu)}}\frac{1-e^{i\ell(\lambda_k+\mu)}}{1-e^{i(\lambda_k+\mu)}}g(\mu)\\
&\quad+(\1_{\lambda_j>0}g(\lambda_j)+\1_{\lambda_k<0}g(-\lambda_k))\left(1+\frac{\ell}{L}(e^{-2i\theta}-1) \right)\frac{e^{-2i\theta}-1}{L}e^{i(\lambda_j+\lambda_k)}\frac{1-e^{i\ell(\lambda_j+\lambda_k)}}{1-e^{i(\lambda_j+\lambda_k)}}\,,
\end{aligned}
\end{equation}
while for $\lambda_j=-\lambda_k$ we have
\begin{equation}
\begin{aligned}
\tilde{E}_{jk}(\pmb{\bar{\lambda}})&=\left(\frac{e^{-2i\theta}-1}{L}\right)^2\sum_{\substack{\mu\in{\rm NS}_+\\\mu\neq \lambda_j}}\left| \frac{1-e^{i\ell(\lambda_j-\mu)}}{1-e^{i(\lambda_j-\mu)}}\right|^2g(\mu)\\
&+\1_{\lambda_j>0}\left(1+\frac{\ell}{L}(e^{-2i\theta}-1) \right)^2g(\lambda_j)\,.
\end{aligned}
\end{equation}
The factor $(-1)^{N/2}$ arises from the re-ordering of the columns of the matrix in order to use Lemma \ref{brujin}. Factorizing $g(\lambda)$ for $\lambda>0$, in the thermodynamic limit one
obtains a Fredholm Pfaffian
\begin{equation}
(-1)^{N/2}\pf
  [\tilde{E}(\pmb{\bar{\lambda}})-\tilde{E}(\pmb{\bar{\lambda}})^T]=\Pf[{\rm
      Jd}+\mathcal{E}[\rho]]\prod_{\lambda\in\pmb{\lambda}}g(\lambda)(1+o(L^0))\ .
\label{FCSeq2}
\end{equation}
Here the kernel acts on $[-\pi,\pi]\times [-\pi,\pi]$
\begin{equation}
\begin{aligned}
\mathcal{E}[\rho](\lambda,\mu)&=\frac{(e^{-2i\theta}-1)^2}{2\pi}\frac{\sqrt{\rho(\lambda)\rho(\mu)}}{g^+(\lambda)g^+(\mu)}\\
&\times\int_0^\pi \left[\frac{1-e^{i\ell(\lambda-k)}}{1-e^{i(\lambda-k)}}\frac{1-e^{i\ell(\mu+k)}}{1-e^{i(\mu+k)}}-\frac{1-e^{i\ell(\mu-k)}}{1-e^{i(\mu-k)}}\frac{1-e^{i\ell(\lambda+k)}}{1-e^{i(\lambda+k)}}\right]g(k)\D{k}\\
&+(e^{-2i\theta}-1)\frac{\sqrt{\rho(\lambda)\rho(\mu)}}{g^+(\lambda)g^+(\mu)}\frac{1-e^{i\ell(\lambda+\mu)}}{1-e^{i(\lambda+\mu)}}(g(\lambda)-g(\mu))\,,
\end{aligned}
\end{equation}
the function $\rho(\lambda)$ is the root density associated
  with $\pmb{\lambda}$ and $g^+(\lambda)$ is defined by
\begin{equation}
g^+(\lambda)=\begin{cases}
g(\lambda)\qquad \text{if }\lambda>0\ ,\\
1\qquad \text{if }\lambda<0\ .
\end{cases}
\end{equation}
The factor $\sqrt{\rho(\lambda)\rho(\mu)}$ ensures that in the definition \eqref{pfafdef}, each integral over $[0,\pi]$ comes with a root density factor $\rho(\lambda)$. Substituting \fr{FCSeq2} and \fr{eq1} into \fr{FCSeq1} we obtain
\begin{equation}
\langle e^{i\theta \sum_{j=1}^{\ell}
  \sigma^z_j}\rangle_f=e^{i\theta\ell}|A|^2\sum_{\pmb{\lambda}\subset{\rm
    NS}_+}\Pf[{\rm
    Jd}+\mathcal{E}[\rho]]\prod_{\lambda\in\pmb{\lambda}}|g(\lambda)|^2(1+o(L^0))\ .
\end{equation}
Finally we employ Lemma \ref{qa} to arrive at our final result in
terms of a Fredholm Pfaffian
\begin{equation}
\boxed{\langle e^{i\theta \sum_{j=1}^{\ell} \sigma^z_j}\rangle_f= e^{i\theta\ell} \Pf[{\rm Jd}+\mathcal{E}[\rho_s]]}\,,
\end{equation}
where
\begin{equation}
\rho_s(k)=\frac{1}{2\pi}\frac{|g(k)|^2}{1+|g(k)|^2}\,.
\end{equation}
\subsection{Order-parameter one-point function }
In contrast to $\sigma_\ell^z$ the longitudinal spin operator $\sigma^x_\ell$ is
non-local in the Jordan-Wigner fermions and as a consequence the
computation of its expectation value is a non-trivial problem. In this 
section we present a formula for the expectation value of the
magnetization in the CE as defined in \eqref{ceodd}, i.e.
\be
\langle \sigma^x_\ell\rangle_f\equiv
\Re \left[\Psi^{\rm R}_{h\gamma}(A^{\rm R},f)^\dagger
  \sigma^x_\ell\  \Psi^{\rm NS}_{h\gamma}(A^{\rm NS},f)\right]\ .
\label{OP}
\ee
Since $\sigma^x_\ell$ is an odd operator under fermion parity it
maps NS (R) states onto R (NS) states and only averages like \fr{OP}
are non-vanishing. We note that they arise naturally in the context of
spontaneous symmetry breaking of the spin-flip $\mathbb{Z}_2$ symmetry.
The average \fr{OP} has been derived in the Supplemental Material
of \cite{DBG} in the particular case of the Ising model, and the
generalization to the XY model is straightforward. The 
result takes the form of a Fredhom determinant
\begin{equation}\label{mag}
\boxed{\langle \sigma^x_\ell\rangle_f=\Re \Det[{\rm Id}+\mathcal{M}[\rho_s]]}\,,
\end{equation}
where we defined the following kernel acting on $[0,\pi]\times [0,\pi]$
\begin{align}
\mathcal{M}[\rho](\lambda,\mu)&=-\frac{2}{\pi}\frac{\rho(\lambda)\sin
  \lambda  }{h(\lambda)}\frac{1}{\cos \lambda-\cos
  \mu}\left[\int_0^\pi \frac{h(k)\sin k}{\cos \lambda-\cos
    k}\D{k}-\int_0^\pi \frac{h(k)\sin k}{\cos \mu-\cos
    k}\D{k}\right]\ ,\label{sigmaFD}\\
h(k)&=\frac{iK_{01;h\gamma}(k)+f(k)}{1+iK_{01;h\gamma}(k)f(k)}\ ,
\label{hhh}\qquad
\rho_s(k)=\frac{1}{2\pi}\frac{|h(k)|^2}{1+|h(k)|^2}\,.
\end{align}

In \eqref{mag} we have assumed that the function $h$ in \eqref{hhh} is
regular. We stress that in \fr{OP} the amplitudes $A^{\rm R}$ and
$A^{\rm NS}$ are given by \fr{ARNS}. In applications to
time-dependent ramps the additional phase factor discussed in
section \ref{secpha} needs to be taken into account, \emph{cf.}
section \ref{KZsec}.

\subsection{Equal-time order-parameter two-point function}

The purpose of this section is to derive the static two-point
correlation function
\begin{equation}\label{goal2}
\langle \sigma^x_{\ell+1} \sigma^x_1\rangle_f\,.
\end{equation}
We note that exact Pfaffian/determinant representations of size
$2\ell$ for the order parameter two point function in an arbitrary GE
have been derived before using Wick's theorem and various explicit
  results on large-distance asymptotics have been derived, see
  e.g. Refs \cite{xymodel1,xymodel2,xymodel3,Cherng06,CEF1,CEF2}.

To compute \eqref{goal2}, we express the coherent states in the $(0,1)$ basis and insert a complete set of eigenstates between the two $\sigma^x$ operators to obtain
\begin{equation}
\langle \sigma^x_{\ell+1} \sigma^x_1\rangle_f=|A|^2 \sum_{\pmb{\lambda},\pmb{\mu}\subset{\rm NS}_+}\sum_{\pmb{\nu}\subset {\rm R}} {}^{\rm NS}_{01}\langle \pmb{\bar{\lambda}}| \sigma^x_1|\pmb{\nu}\rangle_{01}^{\rm R} {}^{\rm R}_{01}\langle \pmb{\nu}| \sigma^x_1|\pmb{\bar{\mu}}\rangle_{01}^{\rm NS}\prod_{\lambda\in\pmb{\lambda}}h^*(\lambda)\prod_{\mu\in\pmb{\mu}}h(\mu)\prod_{\nu\in\pmb{\nu}}e^{i\ell\nu}\,,
\end{equation}
with $h(k)$ defined as in \eqref{hhh}. Using Lemma \ref{ffsigmax} to
express the form factor of $\sigma^x$ as a determinant, and Lemma
\ref{andreief} to sum over $\pmb{\nu}$, we obtain 
\begin{equation}
\langle \sigma^x_{\ell+1} \sigma^x_1\rangle_f=|A|^2 \sum_{\pmb{\lambda},\pmb{\mu}\subset{\rm NS}_+} \det C(\pmb{\bar{\lambda}},\pmb{\bar{\mu}})\prod_{\lambda\in\pmb{\lambda}}h^*(\lambda)\prod_{\mu\in\pmb{\mu}}h(\mu)\,,
\end{equation}
where
\begin{equation}\label{cc}
\begin{aligned}
C(\pmb{p},\pmb{q})_{jk}&=\frac{4}{L^2}\sum_{\nu\in{\rm R}}\frac{e^{i(\ell+1)\nu}}{(e^{ip_j}-e^{i\nu})(e^{i\nu}-e^{iq_k})}\,.
\end{aligned}
\end{equation}
To perform this sum, we now use Lemma \ref{sum}. If $p_j\neq q_k$, we
decompose the summand into partial fractions with respect to
$e^{i\nu}$ and use \eqref{sum1} to carry out the sum over $\nu\in{\rm
  R}$. If $p_j=q_k$ we use the derivative with respect to $z$ of
\eqref{sum1}. 
We obtain
\begin{equation}\label{C}
C(\pmb{p},\pmb{q})_{jk}=\begin{cases}
-\frac{2}{L}\frac{e^{i\ell p_j}-e^{i\ell q_k}}{e^{ip_j}-e^{iq_k}}&
\text{if }p_j\neq q_k\ ,\\
\left(1-\frac{2\ell}{L}\right)e^{ip_j(\ell-1)}& \text{if }p_j=q_k\ .
\end{cases}
\end{equation}
We next use Lemma \ref{brujin} to sum over $\pmb{\mu}$, which gives
\begin{equation}
\langle \sigma^x_{\ell+1} \sigma^x_1\rangle_f=|A|^2
\sum_{\pmb{\lambda}\subset{\rm
    NS}_+}(-1)^{N/2}\pf[\tilde{C}(\pmb{\bar{\lambda}})-\tilde{C}(\pmb{\bar{\lambda}})^T]\prod_{\lambda\in\pmb{\lambda}}h^*(\lambda)\ .
\end{equation}
Here $N$ is the number of roots in $\pmb{\bar{\lambda}}$ and 
\begin{align}
\tilde{C}(\pmb{q})_{jk}&=(1-\delta_{q_j+q_k,0})\tilde{C}_1(q_j,q_k)
+\delta_{q_j+q_k,0}\ \tilde{ C}_2(q_j,q_k)\ ,\nn
\tilde{ C}_1(q_j,q_k)&=
\frac{4}{L^2}\sum_{\substack{p\in{\rm NS}_+\\ p\neq q_j,-q_k}}\frac{e^{i\ell q_j}-e^{i\ell p}}{e^{iq_j}-e^{ip}}\frac{e^{-i\ell p}-e^{i\ell q_k}}{e^{-ip}-e^{iq_k}}h(p)\nn
&-(\1_{q_j>0}h(q_j)+\1_{q_k<0}h(-q_k))\frac{2}{L}\left(1-\frac{2\ell}{L}
\right)\frac{1-e^{i\ell(q_k+q_j)}}{1-e^{i(q_k+q_j)}}\ ,\nn
\tilde{C}_2(q_j,q_k)&=\frac{4}{L^2}\sum_{\substack{p\in{\rm
      NS}_+\\ p\neq q_j}}\left|\frac{e^{i\ell q_j}-e^{i\ell
    p}}{e^{iq_j}-e^{ip}}\right|^2h(p)+\1_{q_j>0}\left(1-\frac{2\ell}{L}
\right)^2h(q_j)\ .
\end{align}
Taking the thermodynamic limit we obtain a Fredholm Pfaffian
\begin{equation}
(-1)^{N/2}\pf[\tilde{C}(\pmb{\bar{\lambda}})-\tilde{C}(\pmb{\bar{\lambda}})^T]=\Pf[{\rm Jd}+\mathcal{C}_2[\rho]]\prod_{\lambda\in\pmb{\lambda}}h(\lambda)\,,
\end{equation}
where $\rho$ is the root density corresponding to
$\pmb{\lambda}$ and where $\mathcal{C}_2[\rho]$ is the following kernel acting on $[-\pi,\pi]\times [-\pi,\pi]$
\begin{equation}
\begin{aligned}
\label{correlation}
\mathcal{C}_2[\rho](\lambda,\mu)&=-2
\frac{\sqrt{\rho(\lambda)\rho(\mu)}}{h^+(\lambda)h^+(\mu)}\Biggl[
\frac{1-e^{i(\lambda+\mu)\ell}}{1-e^{i(\lambda+\mu)}}(h(\lambda)-h(\mu))\nn
&-\int_0^\pi\frac{\D{k}}{\pi}
\Bigl(\frac{1-e^{i\ell(\lambda-k)}}{1-e^{i(\lambda-k)}}\frac{1-e^{i\ell(\mu+k)}}{1-e^{i(\mu+k)}}-\frac{1-e^{i\ell(\mu-k)}}{1-e^{i(\mu-k)}}\frac{1-e^{i\ell(\lambda+k)}}{1-e^{i(\lambda+k)}}\Bigr)h(k)\Biggr],\\
\end{aligned}
\end{equation}
with
\begin{equation}
h^+(\lambda)=\begin{cases}
h(\lambda)\qquad \text{if }\lambda>0\\
1\qquad \text{if }\lambda<0
\end{cases}\,.
\end{equation}
We then employ Lemma \ref{qa} to arrive at our final result
\begin{equation}
\boxed{\langle \sigma^x_{\ell+1} \sigma^x_1\rangle_f=\Pf[{\rm Jd}+\mathcal{C}_2[\rho_s]]}\,,
\end{equation}
where
\begin{equation}
\rho_s(k)=\frac{1}{2\pi}\frac{|h(k)|^2}{1+|h(k)|^2}\,.
\end{equation}

\subsection{Equal-time order-parameter three-point function}
The purpose of this section is to show that the strategy employed for
one and two-point functions can be generalized straightforwardly to
higher-point functions. We consider the particular example of the
  order-parameter three-point function
\begin{equation}
\langle \sigma^x_{\ell_2+\ell_1+1}\sigma^x_{\ell_1+1} \sigma^x_1\rangle_f\,.
\label{3ptfn}
\end{equation}
This operator is odd and non-local in terms of the Jordan-Wigner
fermions and as far as we are aware of there is no known Pfaffian or
determinant representation of \fr{3ptfn} in the thermodynamic limit.

We then follow the same steps as for the two-point function by
expressing the two coherent states in the $(0,1)$ basis and inserting
complete sets of eigenstates between each operator to obtain 
\begin{equation}
\begin{aligned}
\langle \sigma^x_{\ell_2+\ell_1+1}\sigma^x_{\ell_1+1} \sigma^x_1\rangle_f=\Re A^{\rm R *}A^{\rm NS}&\sum_{\substack{\pmb{\lambda}\subset {\rm R}_+\\ \substack{\pmb{\mu}\subset {\rm NS}_+}}}\sum_{\substack{\pmb{\nu}\subset{\rm NS}\\ \pmb{\kappa}\subset{\rm R}}} {}^{\rm R}_{01}\langle \pmb{\bar{\lambda}}| \sigma^x_1|\pmb{\nu}\rangle_{01}^{\rm NS} {}^{\rm NS}_{01}\langle \pmb{\nu}| \sigma^x_1|\pmb{\kappa}\rangle_{01}^{\rm R}{}^{\rm R}_{01}\langle \pmb{\kappa}| \sigma^x_1|\pmb{\bar{\mu}}\rangle_{01}^{\rm NS}\\
&\times \prod_{\lambda\in\pmb{\lambda}}h^*(\lambda)\prod_{\mu\in\pmb{\mu}}h(\mu)\prod_{\nu\in\pmb{\nu}}e^{i\ell_2\nu}\prod_{\kappa\in\pmb{\kappa}}e^{i\ell_1\kappa}\,.
\end{aligned}
\end{equation}
Next we perform the sum over $\pmb{\kappa}$ by employing Lemmas \ref{ffsigmax} and
\ref{andreief} and obtain an analogous expression as in the two-point function case
\begin{equation}
\begin{aligned}
\langle \sigma^x_{\ell_2+\ell_1+1}\sigma^x_{\ell_1+1} \sigma^x_1\rangle_f&=\Re A^{\rm R *}A^{\rm NS}\\
&\times\sum_{\substack{\pmb{\lambda}\subset {\rm
      R}_+\\ \substack{\pmb{\mu}\subset {\rm
        NS}_+}}}\sum_{\substack{\pmb{\nu}\subset{\rm NS}}} {}^{\rm
  R}_{01}\langle \pmb{\bar{\lambda}}|
\sigma^x_1|\pmb{\nu}\rangle_{01}^{\rm NS} \det
C(\pmb{\nu},\pmb{\bar{\mu}})
\prod_{\lambda\in\pmb{\lambda}}h^*(\lambda)\prod_{\mu\in\pmb{\mu}}h(\mu)\prod_{\nu\in\pmb{\nu}}e^{i(\ell_2-1/2)\nu}\ .
\end{aligned}
\end{equation}
Here $C(\pmb{\nu},\pmb{\bar{\mu}})$ is given by \eqref{C} with
$\ell$ replaced by $\ell_1$. Then we use Lemmas \ref{ffsigmax}
and \ref{andreief} to perform the sum over $\pmb{\nu}$ and obtain 
\begin{equation}
\begin{aligned}
\langle \sigma^x_{\ell_2+\ell_1+1}\sigma^x_{\ell_1+1} \sigma^x_1\rangle_f=\Re A^{\rm R *}A^{\rm NS}&\sum_{\substack{\pmb{\lambda}\subset {\rm R}_+\\ \substack{\pmb{\mu}\subset {\rm NS}_+}}}  \det  C'(\pmb{\bar{\lambda}},\pmb{\bar{\mu}}) \prod_{\lambda\in\pmb{\lambda}}h^*(\lambda)\prod_{\mu\in\pmb{\mu}}h(\mu)\,,
\end{aligned}
\end{equation}
where
\begin{equation}\label{cprime}
C'(\pmb{p},\pmb{q})=\frac{2}{L}\sum_{\nu\in {\rm
    NS}}\frac{e^{i(\ell_2+1)\nu}}{e^{ip_j}-e^{i\nu}}\times\begin{cases} 
-\frac{2}{L}\frac{e^{i\ell_1 \nu}-e^{i\ell_1q_k}}{e^{i \nu}-e^{iq_k}}&
\text{if }\nu\neq q_k\ ,\\
(1-\tfrac{2\ell_1}{L})e^{i\nu(\ell_1-1)}& \text{if }\nu= q_k\ .
\end{cases}
\end{equation}
Writing
\begin{equation}
\frac{e^{i\ell_1 \nu}-e^{i\ell_1q_k}}{e^{i \nu}-e^{iq_k}}=e^{i(\ell_1-1)\nu}\sum_{m=0}^{\ell_1-1}e^{im(q_k-\nu)}\,,
\end{equation}
we can use Eq \eqref{sum2} in Lemma \ref{sum} to compute $C'(\pmb{p},\pmb{q})$. We find
\begin{equation}
C'(\pmb{p},\pmb{q})_{jk}=\frac{2}{L}\frac{e^{i(\ell_1+\ell_2)p_j}-e^{i(\ell_1q_k+\ell_2p_j)}+e^{i(\ell_1+\ell_2)q_k}}{e^{ip_j}-e^{iq_k}}\,.
\end{equation}
We then use Lemma \ref{brujin} to sum over $\pmb{\mu}$ to obtain
\begin{equation}
\label{eq2}
\langle \sigma^x_{\ell_2+\ell_1+1}\sigma^x_{\ell_1+1} \sigma^x_1\rangle_f=\Re A^{\rm R *}A^{\rm NS} \sum_{\pmb{\lambda}\subset{\rm R}_+}(-1)^{N/2}\pf[\tilde{C'}(\pmb{\bar{\lambda}})-\tilde{C'}(\pmb{\bar{\lambda}})^T]\prod_{\lambda\in\pmb{\lambda}}h^*(\lambda)\,,
\end{equation}
where $N$ is the number of roots in $\pmb{\bar{\lambda}}$ and
\begin{equation}
\begin{aligned}
\tilde{C'}(\pmb{\bar{\lambda}})_{jk}=\frac{4}{L^2}\sum_{q\in{\rm NS}_+}&\frac{e^{i(\ell_1+\ell_2)p_j}-e^{i(\ell_1q+\ell_2p_j)}+e^{i(\ell_1+\ell_2)q}}{e^{ip_j}-e^{iq}}\\
&\times\frac{e^{i(\ell_1+\ell_2)p_k}-e^{i(-\ell_1q+\ell_2p_k)}+e^{-i(\ell_1+\ell_2)q}}{e^{ip_k}-e^{-iq}}h(q)\,.
\end{aligned}
\end{equation}
In the thermodynamic limit  the remaining sum can be converted
into an integral, except when $p_j=-p_q$ where an additional
contribution $\delta_{p_j,-p_k}h(p_j)$ arises from the double
pole in $q$. This results in a Fredholm Pfaffian  
\begin{equation}
(-1)^{N/2}\pf[\tilde{C'}(\pmb{\bar{\lambda}})-\tilde{C'}(\pmb{\bar{\lambda}})^T]=\Pf[{\rm
      Jd}+\mathcal{C}_3[\rho]]\prod_{\lambda\in\pmb{\lambda}}h(\lambda)
+{\cal O}(L^{-1})\ ,
\end{equation}
where $\rho$ is the root density corresponding to $\pmb{\lambda}$ and where $\mathcal{C}_3[\rho]$ is the following kernel acting on $[-\pi,\pi]\times [-\pi,\pi]$
\begin{align}
\mathcal{C}_3[\rho](\lambda,\mu)&=\frac{2}{\pi}\frac{\sqrt{\rho(\lambda)\rho(\mu)}}{h^+(\lambda)h^+(\mu)}\int_0^\pi
\Big[a(\lambda,k)\ a(\mu,-k)-a(\lambda,-k)\ a(\mu,k)\Big]h(k)\D{k}\ ,\nn
a(\lambda,k)&=
\frac{1-e^{i\ell_2(\lambda-k)}+e^{i(\ell_1+\ell_2)(\lambda-k)}}{e^{i\lambda}-e^{ik}}.
\end{align}
This expression for $\mathcal{C}_3[\rho](\lambda,\mu)$ is to be
  understood as a principal value integral with simple poles at
$k=\pm\lambda,\pm\mu$ for $\lambda\neq-\mu$, and is defined by
continuity for $\lambda=-\mu$. Finally we apply Lemma \ref{qa} to
\eqref{eq2}, which results in the Fredholm Pfaffian
\begin{equation}\label{3}
\boxed{\langle \sigma^x_{\ell_2+\ell_1+1}\sigma^x_{\ell_1+1} \sigma^x_1\rangle_f=\Re \Pf[{\rm Jd}+\mathcal{C}_3[\rho_s]]}\,,
\end{equation}
where
\begin{equation}
\rho_s(k)=\frac{1}{2\pi}\frac{|h(k)|^2}{1+|h(k)|^2}\,.
\end{equation}
In \eqref{3} we have once again assumed that the function $h$ in \eqref{hhh} is
regular. 

\subsection{Dynamical order-parameter two-point function}
We now turn to the non-equal-time two-point function of $\sigma^x$ in
the CE, i.e.
\begin{equation}
C^{xx}(\ell,t)
\equiv\langle \sigma^x_{\ell+1}(t/2)\sigma^x_1(-t/2) \rangle_f
=\langle e^{itH(h,\gamma)/2} \sigma^x_{\ell+1} e^{-itH(h,\gamma)}
\sigma^x_1 e^{itH(h,\gamma)/2} \rangle_f\ .
\label{dyn}
\end{equation}
A particular case of the correlator \fr{dyn} is the dynamical
two-point function in the XY model in a field after a quantum
quench. This has been considered previously for $\gamma=1$ and
analytic results were obtained at low densities of excitations and
large space/time separations \cite{essler12}.

\subsubsection{Summation of the $\sigma^x$ form factors}

Without loss of generality we choose the coherent state in \eqref{dyn}
to belong to the R sector and then expand it as \eqref{coherent} in the $(0,1)$ basis. We
then insert a complete set of eigenstates between each of the
operators to obtain
\begin{align}\label{todod}
C^{xx}(\ell,t)&=A^{\rm
  R}_tA^{\rm R*}_{-t}
\sum_{\pmb{q},\pmb{k}\subset{\rm
    R}_+}\sum_{\pmb{\lambda},\pmb{\mu}\subset{\rm NS}} \left[
  \prod_{q\in\pmb{q}}h^*_{-t}(q)\prod_{k\in\pmb{k}}h_t(k)\right]\nn
&\qquad\ \times\ {}^{\rm R}_{01}\langle
\pmb{\bar{\bar{q}}}|\sigma^x_{\ell+1} |\pmb{\lambda}\rangle^{\rm
  NS}_{01}\ {}^{\rm NS}_{01}\langle \pmb{\lambda}
|e^{-iH(h,\gamma)t}|\pmb{\mu}\rangle^{\rm NS}_{01}\ {}^{\rm
  NS}_{01}\langle
\pmb{\mu}|\sigma^x_{1}|\pmb{\bar{\bar{k}}}\rangle_{01}^{\rm R}\,, 
\end{align}
where we have from Theorem \ref{thm} and Eq \eqref{evoty}
\begin{align}
\label{hdyna}
h_t(k)&=\frac{iK_{01;h\gamma}(k)+e^{it\varepsilon_{h\gamma}(k)}f(k)}{1+iK_{01;h\gamma}(k)e^{it\varepsilon_{h\gamma}(k)}f(k)}\ ,\\
\label{ar}
A_t^{\rm R}&=e^{it \mathfrak{E}^{\rm R}/2}\prod_{k\in {\rm
    R}_+}\sqrt{\frac{1+K_{01;h\gamma}(k)^2}{1+|f(k)|^2}}
\frac{1}{1-iK_{01;h\gamma}(k) h_t(k)}\ .
\end{align}
For later convenience we introduce
\begin{equation}\label{ans}
A_t^{\rm NS}=e^{it \mathfrak{E}^{\rm NS}/2}\prod_{k\in {\rm NS}_+}\sqrt{\frac{1+K_{01;h\gamma}(k)^2}{1+|f(k)|^2}} \frac{1}{1-iK_{01;h\gamma}(k) h_t(k)}\,.
\end{equation}
In the remainder of the section we will use the shorthand notations
$K(k)\equiv K_{01;h\gamma}(k)$ and $\varepsilon(k)\equiv
\varepsilon_{h\gamma}(k)$. 

To evaluate \eqref{todod}, we first express the $\sigma^x$ form
factors as determinants using Lemma \ref{ffsigmax}. Because of the pair structure of the states $\pmb{\bar{\bar{k}}}$ and $\pmb{\bar{\bar{q}}}$, each $k_i$ and $q_j$ appear twice in these determinants. Hence the sums over $\pmb{q},\pmb{k}\subset {\rm R}_+$ are of the form of Lemma \ref{brujin}. It yields
\begin{align}
\label{eqlemma}
C^{xx}(\ell,t)&=
A_t^{\rm R}A^{\rm R*}_{-t}\sum_{\pmb{\lambda},\pmb{\mu}\subset{\rm
    NS}} {}^{\rm NS}_{01}\langle \pmb{\lambda}
|e^{-iH(h,\gamma)t}|\pmb{\mu}\rangle^{\rm
  NS}_{01}\ \pf[D_{t}(\pmb{\mu})-D_{t}(\pmb{\mu})^T]\nn
&\qquad\qquad\times\ \pf[D_{-t}(\pmb{\lambda})-D_{-t}(\pmb{\lambda})^T]^*\prod_{\lambda\in\pmb{\lambda}}e^{i(\ell+1/2) \lambda}\prod_{\mu\in\pmb{\mu}}e^{-i\mu/2}\,,
\end{align}
where
\begin{equation}
D_{t}(\pmb{\mu})_{jk}=\frac{4}{L^2}\sum_{p\in {\rm R}_+}\frac{h_t(p)}{(e^{ip}-e^{i\mu_j})(e^{-ip}-e^{i\mu_k})}\,.
\end{equation}
The thermodynamic limit of this expression is
\begin{equation}\label{ddd}
\begin{aligned}
D_t(\pmb{\mu})_{jk}&=h_t(\mu_j)\delta_{\mu_j,-\mu_k}\1_{\mu_j>0}\\
&+\frac{2}{\pi L(1-e^{i(\mu_j+\mu_k)})}\left[\int_0^\pi \frac{h_t(p)}{1-e^{i(\mu_j-p)}}\D{p}-\int_0^\pi \frac{h_t(p)}{1-e^{-i(\mu_k+p)}}\D{p}\right]+\mathcal{O}(L^{-2})\,,
\end{aligned}
\end{equation}
where the second term is understood as a derivative when
$\mu_j=-\mu_k$. 

\subsubsection{Thermodynamic limit of the Pfaffians}
The thermodynamic limit of the Pfaffians appearing in \eqref{eqlemma}
is more involved than for the equal-time correlations treated in the
previous sections. Indeed, $\pmb{\lambda}$ and $\pmb{\mu}$ are not
necessarily pair states and so the ``anti-diagonal'' term
$\delta_{\mu_j,-\mu_k}\1_{\mu_j>0}$ in \eqref{ddd} is not always
present. To treat this complication we introduce two sets of momenta 
$\pi(\pmb{\mu}),\sigma(\pmb{\mu})$ as in Lemma \ref{overlap}. One sees
that the behaviour of $D_t(\pmb{\mu})-D_t(\pmb{\mu})^T$ significantly depends on whether the
$\mu$'s are paired $\mu\in\pi(\pmb{\mu})$, in which case there is a
non-zero anti-diagonal term $\delta_{\mu_j,-\mu_k}$ of order $L^0$, or
whether they are not paired $\mu\in\sigma(\pmb{\mu})$, in which case this
``anti-diagonal'' term is absent. In order to use Lemma \ref{fredh} we employ
Cayley's relation 
\begin{equation}\label{cayley}
\pf[D_{t}(\pmb{\mu})-D_{t}(\pmb{\mu})^T]^2=\det[D_{t}(\pmb{\mu})-D_{t}(\pmb{\mu})^T]\,,
\end{equation}
and write
\begin{align}
[D_{t}(\pmb{\mu})-D_{t}(\pmb{\mu})^T]_{jk}=&h_t(\mu_j)\delta_{\mu_j,-\mu_k}+\frac{1}{L}d_t(\mu_j,\mu_k)\,,\nn
d_t(\lambda,\mu)=&\frac{2}{\pi (1-e^{i(\lambda+\mu)})}\Bigg[\int_{0}^\pi \frac{h_t(p)}{1-e^{i(\lambda-p)}}\D{p}+\int_{0}^\pi \frac{h_t(p)}{1-e^{-i(\lambda+p)}}\D{p}\nn
&\hskip2.2cm-\int_{0}^\pi \frac{h_t(p)}{1-e^{i(\mu-p)}}\D{p}-\int_{0}^\pi \frac{h_t(p)}{1-e^{-i(\mu+p)}}\D{p}\Bigg]\,.
\end{align}
In the determinant \eqref{cayley} we then rearrange the lines and columns
in such a way that the paired roots $\mu_j\in\pi(\pmb{\mu})$ appear on
the ``anti-diagonal'' of the matrix $D_{t}(\pmb{\mu})-D_{t}(\pmb{\mu})^T$
and are ordered among themselves (but the unpaired roots in
$\sigma(\pmb{\mu})$ are not necessarily ordered). We then
factorize $D_{t}(\pmb{\mu})-D_{t}(\pmb{\mu})^T=LR$, where
\be
R_{ij}=\delta_{i,N+1-j}\sign(j-i)\prod_{\mu\in\pi(\pmb{\mu})}
h_t^2(\mu)\ ,\quad i,j=1,\dots,N.
\ee
This
way, the determinant  $\det[D_{t}(\pmb{\mu})-D_{t}(\pmb{\mu})^T]$ is
of the form of Lemma \ref{fredh}, with $n$ the 
number of unpaired roots $\sigma(\pmb{\mu})=\{\nu_1,...,\nu_n\}$ and
with functions
\begin{align}
f(\lambda,\mu)&=\frac{d_t(\lambda,-\mu)}{h^+_t(\lambda)h^+_t(\mu)}\ ,\quad
g_j(\lambda)=\frac{d_t(\lambda,\nu_j)}{h^+_t(\lambda)}
\ ,\quad h_i(\mu)=\frac{d_t(\nu_i,-\mu)}{h^+_t(\mu)}\ ,\quad
a_{i,j}=d_t(\nu_i,\nu_j)\,,
\end{align}
where we introduced
\begin{equation}
h_t^+(\lambda)=\begin{cases}
h_t(\lambda)& \text{if }\lambda>0\ ,\\
1& \text{if }\lambda<0\ .
\end{cases}
\end{equation}
We thus obtain as we approach the thermodynamic limit
\begin{align}
\det[D_t(\pmb{\mu})-D_t(\pmb{\mu})^T]&=\frac{1}{L^{|\sigma(\pmb{\mu})|}}\Det_{\lambda,\mu}[{\rm
    Id}+\mathcal{D}_{\rho,t}(\lambda,-\mu)] \det\left[
  \mathcal{F}_{\rho,t}(\nu_i,\nu_j)
  \right]_{\nu_i,\nu_j\in\sigma(\pmb{\mu})}\nn
&\qquad\times\prod_{\mu\in\pi(\pmb{\mu})}h_t^2(\mu)\ ,
\end{align}
where $\rho$ the root density corresponding to $\pmb{\mu}$,
and where we defined the following kernel acting on $[-\pi,\pi]\times
[-\pi,\pi]$ 
\begin{equation}
\begin{aligned}
\mathcal{D}_{\rho,t}(\lambda,\mu)&=\frac{\sqrt{\rho(\lambda)\rho(\mu)}}{h_t^+(\lambda)h_t^+(\mu)}d_t(\lambda,\mu)\,,
\end{aligned}
\end{equation}
and $\mathcal{F}_{\rho,t}(\lambda,\mu)$ satisfies the linear integral equation
\begin{equation}\label{F}
\mathcal{F}_{\rho,t}(\lambda,\mu)+\int_{-\pi}^\pi\frac{d_t(\lambda,-\nu)}{h_t^+(\lambda)h_t^+(\nu)}\mathcal{F}_{\rho,t}(\nu,\mu)\rho(\nu)\D{\nu}=d_t(\lambda,\mu)\ .
\end{equation}
Eqn \fr{F} is obtained from \eqref{reslem10} by using the
equation for the resolvent \eqref{reso1} as well as its equivalent
definition \eqref{reso2}. It is useful to define two further functions
by 
\begin{align}
\mathcal{D}'_{\rho,t}(\lambda,\mu)&=\frac{\sqrt{\rho(\lambda)\rho(\mu)}}{(h_t^+)^*(\lambda)(h_t^+)^*(\mu)}d_t^*(\lambda,\mu)\ ,\nn
\mathcal{F}'_{\rho,t}(\lambda,\mu)&+\int_{-\pi}^\pi\frac{d^*_t(\lambda,-\nu)}{(h_t^+)^*(\lambda)(h_t^+)^*(\nu)}\mathcal{F}'_{\rho,t}(\nu,\mu)\rho(\nu)\D{\nu}=d^*_t(\lambda,\mu)\ .
\end{align}
Now, using that $d_t(\lambda,\mu)=-d_t(\mu,\lambda)$, we find from \eqref{F} and \eqref{reso2} that $\mathcal{F}_{\rho,t}(\lambda,\mu)=-\mathcal{F}_{\rho,t}(\mu,\lambda)$. Hence $\left[ \mathcal{F}_{\rho,t}(\nu_i,\nu_j)
  \right]_{\nu_i,\nu_j\in\sigma(\pmb{\mu})}$ is an antisymmetric
matrix, and we can write its determinant as the square of its
Pfaffian 
\begin{equation}
\det\left[ \mathcal{F}_{\rho,t}(\nu_i,\nu_j) \right]_{\nu_i,\nu_j\in\sigma(\pmb{\mu})}=\pf\left[ \mathcal{F}_{\rho,t}(\nu_i,\nu_j) \right]_{\nu_i,\nu_j\in\sigma(\pmb{\mu})}^2\,.
\end{equation}
This results in a Fredholm Pfaffian as we approach the
thermodynamic limit
\begin{equation}
(-1)^{N/2}\pf[D_t(\pmb{\mu})-D_t(\pmb{\mu})^T]=\frac{1}{L^{|\sigma(\pmb{\mu})|/2}}\Pf[{\rm Jd}+\mathcal{D}_{\rho,t}] \pf\left[ \mathcal{F}_{\rho,t}(\nu_i,\nu_j) \right]_{\nu_i,\nu_j\in\sigma(\pmb{\mu})}\prod_{\mu\in\pi(\pmb{\mu})}h_t(\mu)\,.
\end{equation}

\subsubsection{Summation over the $e^{-itH(h,\gamma)}$ form factors}
Returning to \eqref{eqlemma} we now see that the form factor of
$e^{-itH(h,\gamma)}$ given in Lemma \ref{ffh} imposes that
$\sigma(\pmb{\lambda})=\sigma(\pmb{\mu})$. This permits us to write 
\begin{align}
\label{eqp}
C^{xx}(\ell,t)&=
A^{\rm R}_tA^{\rm R*}_{-t}\sum_{\substack{\pmb{\nu}\subset {\rm NS}\\ \pmb{\nu}\cap (-\pmb{\nu})=\emptyset}}\frac{\prod_{\nu\in\pmb{\nu}}e^{i\ell \nu}}{L^{|\pmb{\nu}|}}\!\!
\sum_{\substack{\pmb{\lambda},\pmb{\mu}\subset\\{\rm
      NS}_+-\{\pmb{\nu},-\pmb{\nu}\}}} {}^{\rm NS}_{01}\langle
\pmb{\bar{\lambda}}\cup \pmb{\nu}
|e^{-iH(h,\gamma)t}|\pmb{\bar{\mu}}\cup \pmb{\nu}\rangle^{\rm NS}_{01}\nn
&\times
\prod_{\lambda\in\pmb{\lambda}}h_{-t}^*(\lambda)\prod_{\mu\in\pmb{\mu}}h_t(\mu)\ 
\Pf[{\rm Jd}+\mathcal{D}_{\rho,t}]\
\Pf[{\rm Jd}+\mathcal{D}_{\rho',-t}']\nn
&\times\pf \left[ \mathcal{F}_{\rho,t}(\nu_i,\nu_j) \right]_{\nu_i,\nu_j\in\pmb{\nu}}\pf\left[ \mathcal{F}'_{\rho',-t}(\nu_i,\nu_j) \right]_{\nu_i,\nu_j\in\pmb{\nu}}\,,
\end{align}
where $\rho$ and $\rho'$ are the root densities corresponding to
$\pmb{\mu}$ and $\pmb{\lambda}$ respectively. 
At fixed $\pmb{\nu}$, given the form
factor of $e^{-itH(h,\gamma)}$  in Lemma \ref{ffh}, the summand is of
the form of Lemma \ref{qa2} with 
\be
f=iK h_{-t}^*\frac{1-e^{-2it\varepsilon}}{1+K^2e^{-2it\varepsilon}}\ , \
g=-iK h_{t}\frac{1-e^{-2it\varepsilon}}{1+K^2e^{-2it\varepsilon}}\ ,\
h=\frac{(1+e^{-2it\varepsilon}K^2)(1+e^{-2it\varepsilon}/K^2)}{(1-e^{-2it\varepsilon})^2},
\ee
and with ${\rm NS}_+$ replaced by ${\rm NS}_+-\{\pmb{\nu}\cup -\pmb{\nu}\}$. To apply Lemma \ref{qa2}, let us first investigate the denominator of Eq \eqref{deno}. We note that the form factor of $e^{-iH(h,\gamma)t}$ in
Lemma \ref{ffh} generates a factor 
\begin{equation}
e^{-it\mathfrak{E}^{\rm NS}}\prod_{k\in{\rm NS}_+}\frac{1+K^2(k)e^{-2it\varepsilon(k)}}{1+K^2(k)}\,.
\end{equation}
Moreover, from \eqref{ans} we have
\begin{equation}
\begin{aligned}
&A^{\rm NS}_{t}A^{\rm NS*}_{-t}e^{-it\mathfrak{E}^{\rm NS}}\prod_{k\in{\rm NS}_+}\frac{1+K^2(k)e^{-2it\varepsilon(k)}}{1+K^2(k)}\\
&=\prod_{k\in{\rm NS}_+}\frac{1+K^2(k)e^{-2it\varepsilon(k)}}{[1+h_th_{-t}^*K^2+(h_th_{-t}^*+K^2)e^{-2it\varepsilon}-iK(h_t-h_{-t}^*)(1-e^{-2it\varepsilon})](k)}\,,
\end{aligned}
\end{equation}
which is precisely the inverse of the denominator in
\eqref{deno}. Hence, defining the following (complex) root density 
\begin{equation}\label{rhot}
\begin{aligned}
\rho_t&=\frac{1}{2\pi}\frac{-iK(1-e^{-2it\varepsilon})h_t+(K^2+e^{-2it\varepsilon})h_th_{-t}^*}{1+h_th_{-t}^*K^2+(h_th_{-t}^*+K^2)e^{-2it\varepsilon}-iK(h_t-h_{-t}^*)(1-e^{-2it\varepsilon})}\ ,\\ 
\rho_t'&=[\rho_{-t}]^*\,,
\end{aligned}
\end{equation}
appearing in  Lemma \ref{qa2}, we obtain
\begin{align}
\label{eqp2}
C^{xx}(\ell,t)=\phi_\infty(t)\phi_\infty(-t)^*&\sum_{\substack{\pmb{\nu}\subset {\rm NS}\\ \pmb{\nu}\cap (-\pmb{\nu})=\emptyset}}\frac{\prod_{\nu\in\pmb{\nu}}s_{\ell,t}(\nu)}{L^{|\pmb{\nu}|}}\
\Pf[{\rm Jd}+\mathcal{D}_{\rho_t,t}]\Pf[{\rm
    Jd}+\mathcal{D}_{\rho_{t}',-t}']\nn
&\times\ \pf \left[ \mathcal{F}_{\rho_t,t}(\nu_i,\nu_j) \right]_{\nu_i,\nu_j\in\pmb{\nu}}\pf\left[ \mathcal{F}'_{\rho_{t}',-t}(\nu_i,\nu_j) \right]_{\nu_i,\nu_j\in\pmb{\nu}}\,,
\end{align}
where
\begin{equation}\label{slt}
s_{\ell,t}(z)=\frac{1}{2\pi}\frac{[1+K^2(z)]e^{i(\ell
  z-t\varepsilon(z))}}{[1+h_th_{-t}^*K^2+(h_th_{-t}^*+K^2)e^{-2it\varepsilon}-iK(h_t-h_{-t}^*)(1-e^{-2it\varepsilon})](z)}
\,,
\end{equation}
and
\begin{equation}
\phi_\infty(t)=\underset{L\to\infty}{\lim}\,\frac{A^{\rm R}_t}{A^{\rm NS}_t}\,.
\end{equation}
The factor $s_{\ell,t}(z)$ arises from the terms in \eqref{ffheq} corresponding to the unpaired roots $\pmb{\nu}$, and the fact that Lemma \ref{qa2} is applied with ${\rm NS}_+$ replaced by ${\rm NS}_+-\{\pmb{\nu}\cup -\pmb{\nu}\}$. The phase $\phi_\infty(t)$ is identical to the phase discussed in Section \ref{secpha}. However, since here the operators involved in the dynamical correlation are time-evolved with a Hamiltonian with a constant magnetic field and anisotropy, the phase can be expressed only in terms of quantities at $t$. It can be straightforwardly computed with Eqs \eqref{ans} and \eqref{ar}, and so is much easier to evaluate numerically than the generic phase discussed in Section \ref{secpha}.

\subsubsection{Representation as a product of Pfaffians}

In the thermodynamic limit the sums over the unpaired roots
$\nu\in\pmb{\nu}$ in \eqref{eqp2} can be converted into
$|\pmb{\nu}|=2n$-fold integrals over $[-\pi,\pi]$, because the cases
where $\nu_i=-\nu_j$ that are excluded in \eqref{eqp2} are negligible
by at least a factor of $L$. This provides us with a
  multiple-integral representation of the form
\begin{align}
\label{dynamical}
C^{xx}(\ell,t)&=
\phi_\infty(t)\phi_\infty(-t)^*\Pf[{\rm Jd}+\mathcal{D}_{\rho_t,t}]\Pf[{\rm Jd}+\mathcal{D}'_{\rho^*_{-t},-t}]\nn
&\times\sum_{n\geq 0}\frac{1}{(2n)!}\int_{-\pi}^\pi \dots \int_{-\pi}^\pi \prod_{j=1}^{2n}s_{\ell,t}(z_j)\underset{i,j}{\pf}[\mathcal{F}_{\rho_t,t}(z_i,z_j)]\underset{i,j}{\pf}[\mathcal{F}'_{\rho_{-t}^*,-t}(z_i,z_j)]\D{z_1}...\D{z_{2n}}\,,
\end{align}
with the convention that the term for $n=0$ in the series is equal to $1$. We now observe that
\begin{equation}
\begin{aligned}
 &\underset{i,j}{\pf}[\mathcal{F}_{\rho_t,t}(z_i,z_j)]\underset{i,j}{\pf}[\mathcal{F}'_{\rho_{-t}^*,-t}(z_i,z_j)]\prod_{j=1}^{2n}s_{\ell,t}(z_j)\\
 &\qquad=\pf\left[\begin{matrix}(s_{\ell,t}(z_i)s_{\ell,t}(z_j)\mathcal{F}_{\rho_t,t}(z_i,z_j))_{1\leq i,j\leq 2n}&0\\0&(\mathcal{F}'_{\rho_{-t}^*,-t}(z_i,z_j))_{1\leq i,j\leq2n}\end{matrix} \right]\\
 &\qquad=\underset{1\leq i,j\leq 2n}{\pf}\left[\begin{matrix}s_{\ell,t}(z_i)s_{\ell,t}(z_j)\mathcal{F}_{\rho_t,t}(z_i,z_j)&0\\0&\mathcal{F}'_{\rho_{-t}^*,-t}(z_j,z_i)\end{matrix} \right]\,.
\end{aligned}
\label{observation}
\end{equation}
The swap of the arguments of $\mathcal{F}'_{\rho_{-t}^*,-t}$ in
the last line compensates the sign factor $(-1)^n$ that results from
changing the $2n\times 2n$ block $2\times 2$ matrix into a $2\times 2$
block $2n\times 2n$ matrix. Eqn \fr{observation} allows us to recast
the series over $n$ of \eqref{dynamical} in the form of the block Fredholm Pfaffian
\eqref{pfa} with $n$ restricted to be even, namely
\begin{equation}
\begin{aligned}
\sum_{n\geq 0}\frac{1}{(2n)!}&\int_{-\pi}^\pi \dots \int_{-\pi}^\pi \prod_{j=1}^{2n}s_{\ell,t}(z_j)\underset{i,j}{\pf}[\mathcal{F}_{\rho_t,t}(z_i,z_j)]\underset{i,j}{\pf}[\mathcal{F}'_{\rho_{-t}^*,-t}(z_i,z_j)]\D{z_1}...\D{z_{2n}}\\
&=\frac{1}{2}\left(\Pf[\pmb{{\rm Jd}}+\pmb{F}_{\ell,t}]+\Pf[\pmb{{\rm Jd}}-\pmb{F}_{\ell,t}] \right)\,,
\end{aligned}
\end{equation}
where $\pmb{F}_{\ell,t}(x,y)$ is a $2\times 2$ matrix-valued function on
$[-\pi,\pi]\times[-\pi,\pi]$ 
\begin{equation}
\pmb{F}_{\ell,t}(x,y)=\left(\begin{matrix}s_{\ell,t}(x)s_{\ell,t}(y)\mathcal{F}_{\rho_t,t}(x,y)&0\\0&\mathcal{F}'_{\rho_{-t}^*,-t}(y,x) \end{matrix} \right)\,.
\end{equation}
From the fact that the $2\times 2$ kernel $\pmb{F}_{\ell,t}$ is diagonal, using that the Pfaffian is multiplied by $-1$ whenever a row and the corresponding columns are multiplied simultaneously by $-1$, we find
\begin{equation}
\Pf[\pmb{{\rm Jd}}-\pmb{F}_{\ell,t}]=\Pf[\pmb{{\rm Jd}}+\pmb{F}_{\ell,t}]\,.
\end{equation}
Hence putting everything
together we arrive at our final result
\begin{equation}\label{dynamicalc}\boxed{
\begin{aligned}
C^{xx}(\ell,t)&=\phi_\infty(t)\phi_\infty(-t)^*&\Pf[{\rm
    Jd}+\mathcal{D}_{\rho_t,t}]\Pf[{\rm
    Jd}+\mathcal{D}'_{\rho^*_{-t},-t}] \
\Pf[\pmb{{\rm Jd}}+\pmb{F}_{\ell,t}]\,.
\end{aligned}}
\end{equation}
Eq \eqref{dynamicalc} again assumes that the functions $h_t$ and
$h_{-t}$ are regular.
\section{Applications \label{app}}
In this section we apply the results reported above in a number of
settings. For simplicity we focus on the case of the transverse-field
Ising chain $\gamma=1$.
\subsection{Order parameter and the Kibble-Zurek mechanism\label{KZsec}}
As a first application we consider the
time-dependence of the order parameter in the transverse-field Ising
model for a ramp of the magnetic field through the
critical point. While this and closely related non-equilibrium
protocols have been previously studied in great detail
\cite{Dziarmaga05,Cherng06,Mukherjee07,Kolodrubetz12,Schuricht16,Hodsagi20,Pollmann2010}
in connection to the Kibble-Zurek mechanism
\cite{Kibble76,Zurek85,Dziarmaga10,Polkovnikov11,delCampo14}, we are
not aware of any results on the dynamics of the order parameter in the
thermodynamic limit. We will consider time-dependent magnetic fields
\begin{equation}\label{cases}
\begin{aligned}
{\rm (i)}\qquad h(t)&=h_0+\alpha t\ ,\qquad\gamma=1\,,\\
{\rm (ii)}\qquad h(t)&=1+\frac{(\alpha t-1+h_0)^3}{(1-h_0)^2}\ ,\qquad\gamma=1\ ,
\end{aligned}
\end{equation}
that cross the critical point linearly (case (i)) or cubically (case (ii)) with a speed parameter $\alpha$, and assume that the system is initialized in the ground state for
$0<h_0<1$ at time $t=0$. The presence of spontaneous symmetry breaking
in the thermodynamic limit can be accounted for by working with the
following initial state, \emph{cf.} Refs \cite{CEF1,CEF2}
\be
|\psi_{h_0}(0)\rangle=\frac{|0\rangle^{\rm R}_{h_01}+\alpha^\dagger_{01;0}|0\rangle^{\rm NS}_{h_01}}{\sqrt{2}}.
\ee
The time evolution of the order parameter in the thermodynamic limit is 
then obtained from \fr{mag}, \fr{phiLoft}
\be\label{resquench}
\langle\psi_{h_0}(t)|\sigma^x_\ell|\psi_{h_0}(t)\rangle=\Re\left( \phi_\infty(t)
\Det[{\rm Id}+\mathcal{M}[\rho_s]]\right)\ .
\ee
Here ${\cal M}$ is given in \fr{sigmaFD} with
\begin{align}
h(k)&=\frac{iK_{01;h(t)1}(k)+f_t(k)}{1+iK_{01;h(t)1}(k)f_t(k)}\ ,\
\end{align}
and $f_t(k)$ is the solution of the nonlinear differential equation
\fr{equadiff}. Importantly the phase factor $\phi_\infty(t)$ is not
always equal to one in this case. We find that there is a sequence of
times $t^*_n$ and associated magnetic fields $h^*_n=h(t^*_n)$ such that
\be
\phi_\infty(t)=(-1)^n\ ,\quad t_n<t<t_{n+1}\ , n=0,1,\dots.
\ee
As an example we consider $h_0=0.6$ and $\alpha=0.3$ for the linear ramp of case (i). Then
\begin{align}
&t^*_0\approx 2.883547\ ,\quad h^*_0\approx 1.465064\ ,\nn
&t^*_1\approx 4.591509\ ,\quad h^*_1\approx 1.977452.
\end{align}
We observe, in agreement with the discussion of Section \ref{secpha}, that this behaviour arises from the
vanishing of the overlaps $\langle\psi_{h_0}(t)|\psi_0(0)\rangle$
at particular times $t^*_n$ in the
thermodynamic limit, which is equivalent to the function $f_t(k)$
becoming singular at times $t^*_n$ for particular wave numbers
$k^*_n$. We note however that the Fredholm determinant appearing in
\eqref{resquench} also exhibits discontinuities at times $t^*_n$ in
such a way that the resulting magnetization
$\langle\sigma^x(t)\rangle$  is a continuous function of time, as
expected on physical grounds. 

The results of a numerical evaluation of the Fredholm determinant expression for the order parameter during the quench are
shown in Fig.~\ref{fig:magn2}. We use \eqref{fredheq} and \cite{Bornemann10} to compute the determinant and use a quadrature
rule with up to 4000 points. We stress that by construction we are
considering the magnetization per site \emph{in the thermodynamic limit}.

\begin{figure}[H]
\begin{center}
 \includegraphics[scale=0.425]{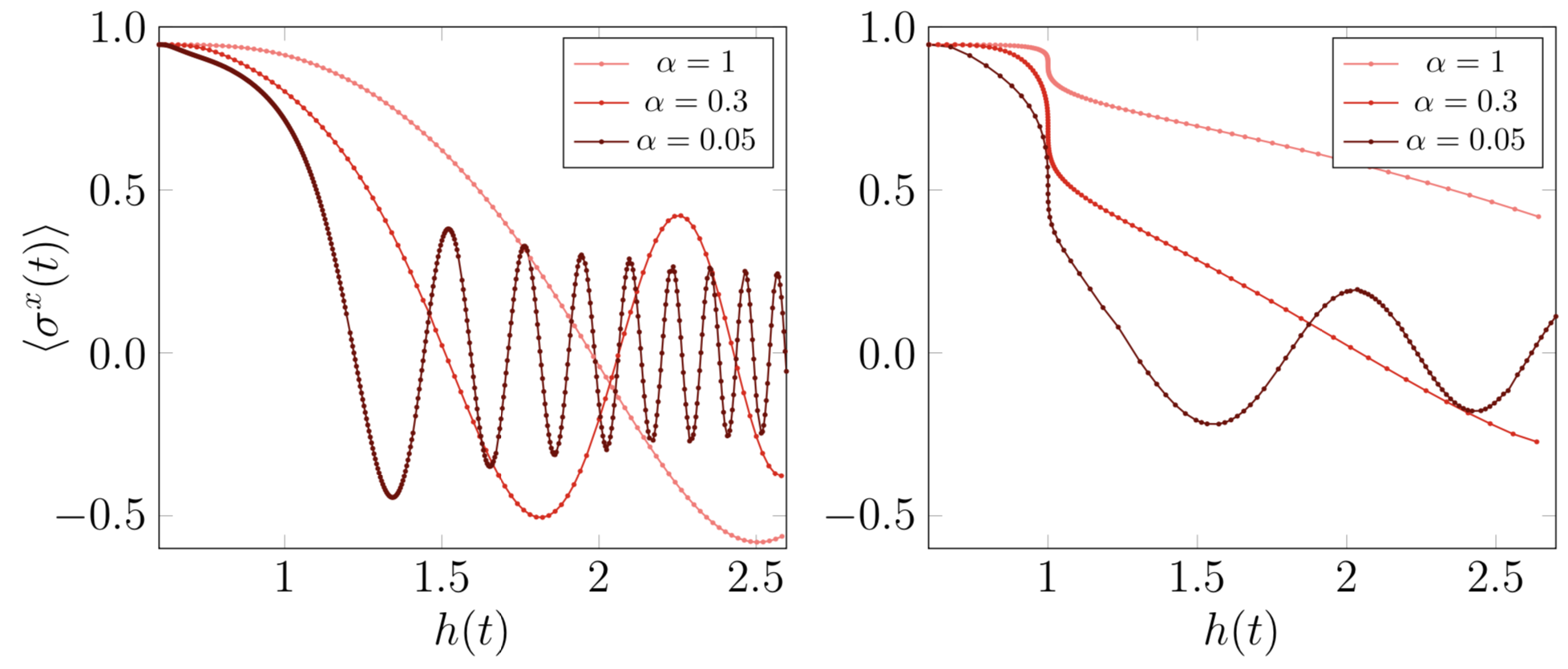}
\end{center}
\caption{Order parameter expectation value $\langle \sigma^x(t)\rangle$ as a function of $h(t)$, with a ramp crossing the critical point linearly $h(t)=0.6+\alpha t$ (left) and cubically $h(t)=1+6.25\times (t\alpha-0.4)^3$ (right).}
\label{fig:magn2}
\end{figure}

We first consider linear ramps across the quantum criticial point
starting in the ordered phase, i.e. case (i) in \eqref{cases} with $h_0$=0.6.
We see that ramping up the magnetic field initially leads to a
reduction of $\langle \sigma^x_j(t)\rangle$, the size of which depends
on the ramp rate $\alpha$. For very fast $\alpha$, 
$\langle \sigma^x_j(t)\rangle$ is expected to remain essentially
pinned to its value at $t=0$: this corresponds to a sudden
approximation and is closely related to the situation encountered in a
quantum quench. A slower ramp rate is expected to result in a faster
reduction of $\langle \sigma^x_j(t)\rangle$ at early times. Both of
these expectations are borne out by the numerical results shown in
Fig.~\ref{fig:magn2}. At later times the magnetization per site
displays an oscillatory behaviour. In the scaling regime around the
critical field $h=1$ this behaviour has been analyzed in some detail
in Ref.~\cite{Hodsagi20}. For a very slow ramp rate the magnetization
closely follows the magnetic field dependence in the ground state, as
expected by the adiabatic approximation, until $h\approx 1$, where
adiabadicity breaks down and Kibble-Zurek physics ensues.
We next consider a nonlinear ramp starting at the same initial field
$h(0)=0.6$ and whose derivative vanishes at the critical point, given by case (ii) in \eqref{cases}. On a very qualitative level the time dependence of the order parameter
is similar to the linear ramp case in that oscillations ensue after an
initial decay.
\subsection{Order parameter in periodically driven systems}
The results derived in the previous sections allow for a
systematic study of the thermodynamic limit of Floquet physics
where the driving magnetic field and anisotropy are periodic functions
of time, see e.g. Refs
\cite{Bukov15,mukherjee09,russomanno12,bastidas12,Lazarides14,eckardt17,oka19,Tapias20,Damski20}.
In Fig.~\ref{fig:floquet} we show the order parameter
$\langle \sigma_\ell^x(t)\rangle$ of a system initialized in the
  ground state at $h=0$ and then driven periodically with frequency $\omega$
\begin{equation}
h(t)=\frac{1-\cos(\pi\omega t)}{2}\ ,\qquad\gamma=1\,.
\end{equation}
\begin{figure}[H]
\begin{center}
 \includegraphics[scale=0.425]{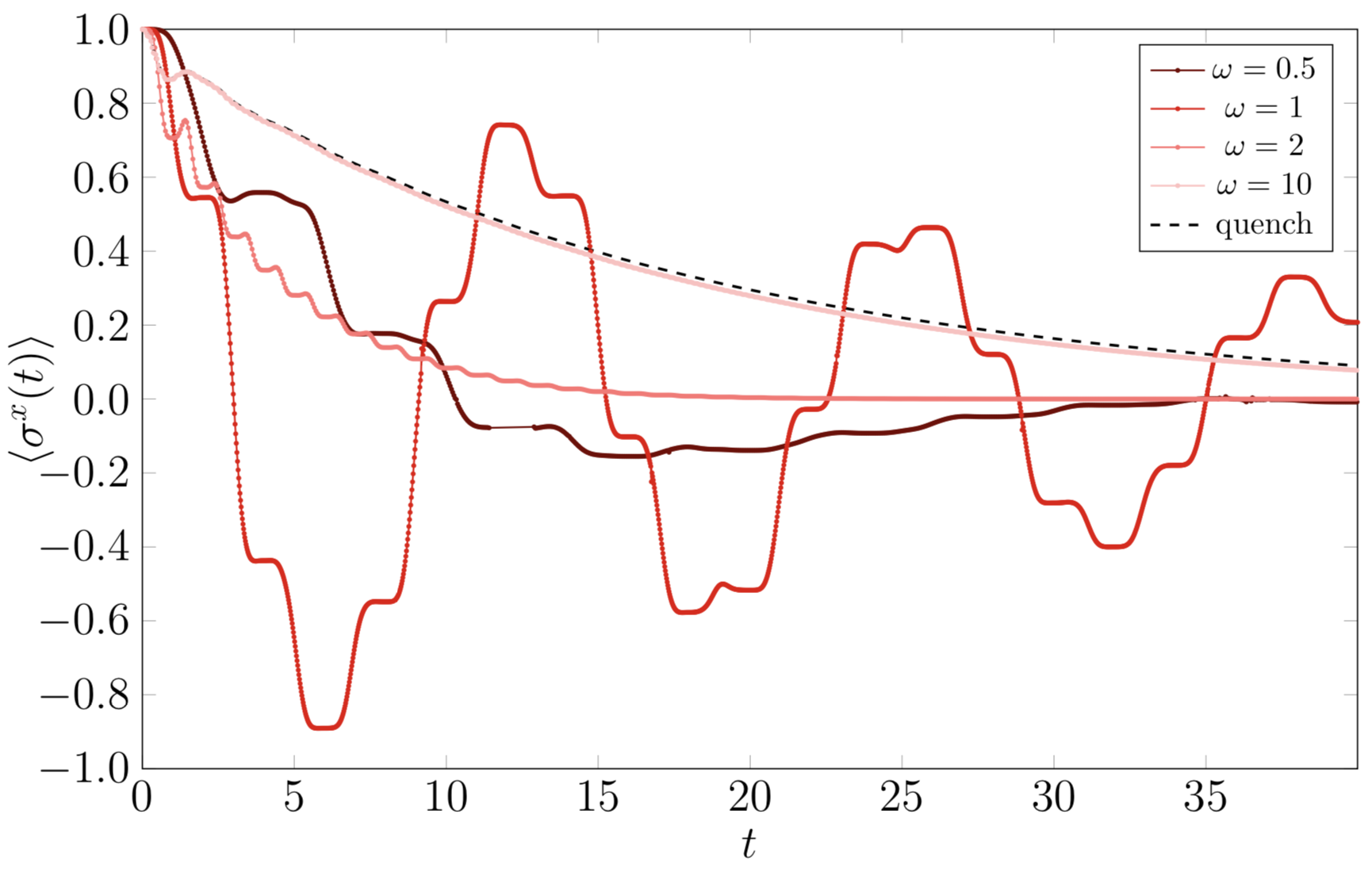}
\end{center}
\caption{Order parameter expectation value $\langle \sigma^x(t)\rangle$ as a function of $ t$, with the variation of magnetic field $h(t)=\tfrac{1}{2}(1-\cos \pi \omega t)$. In dashed is indicated the time-evolution after a sudden quench $h(t)=\frac{1}{2}$ for $t>0$.}
\label{fig:floquet}
\end{figure}
At large frequencies $\omega$ we
expect to recover the results for evolution with the time-averaged
Hamiltonian \cite{Bukov15}, which corresponds to a quantum quench
where the system in initialized in the ground state of $H(0,1)$ and then
time-evolved with $H(\frac{1}{2},1)$. We see that the time evolution
for $\omega=10$ is indeed very close to this limit. At
late times the system synchronizes and can be described by a
``periodic generalized Gibbs ensemble'' \cite{Lazarides14}. In
particular this implies that the order parameter should vanish, which
is indeed what we observe in Fig.~\ref{fig:floquet}.
In the limit of low frequencies $\omega\approx 0$ the behaviour is
initially adiabatic and the order parameter follows the ground state
value at the corresponding magnetic field $h(t)$. As
$t\rightarrow\omega^{-1}$ the magnetic field $h(t)$
approaches its critical value and adiabaticity breaks down and
Kibble-Zurek physics ensues. For frequencies $\omega>2$ the
magnetization is seen to decay towards zero with only weak
oscillations on top of the decay. For frequencies $\omega\approx 1$
there are strong oscillations that decay in time. Interestingly, for
lower frequencies the oscillatory behaviour becomes less pronounced.

\subsection{Dynamical correlations after a sudden quench}

In this section we illustrate formula \eqref{dynamicalc} for dynamical
correlations in a CE for a particular case of the general scenario
discussed in section \ref{outofeq}. We initialize the state 
of the system $|\psi(0)\rangle$ in the ground state of the
transverse field Ising model with magnetic field $h_0$,
i.e. $H(h_0,1)$, and suddenly change the magnetic field to $h$,
triggering a non-trivial time evolution of the state
$|\psi(t)\rangle$. We are then interested in the connected non-equal
time order parameter correlation function
\begin{equation}
\label{goal}
\langle \sigma^x_{\ell+1}(t_1)\sigma^x_1(t_2)\rangle_c=\langle \psi(t_2)|\sigma^x_{\ell+1}e^{i(t_1-t_2)H(h,\gamma)}\sigma^x_{1}|\psi(t_1)\rangle-\langle \sigma^x_{\ell+1}(t_1)\rangle \langle \sigma^x_1(t_2)\rangle\,.
\end{equation}
Numerical results for $\langle
  \sigma^x_{\ell+1}(t_1)\sigma^x_1(t_2)\rangle_c$ for a quench from
$h_0=0.1$ to $h=0.9$, obtained by numerically evaluating
\eqref{dynamicalc}, are shown in Fig.~\ref{ctdensity4}. One sees
that there are several regions where the $2$-point function is
negligibly small. These light-cone structures can be understood
with the quasi-particle picture initially proposed to describe the
growth of entanglement entropy after a quench
\cite{calabrese05,Foini12}. According to this picture, the effect of the 
quench is to create pairs of quasi-particles at each position in the
chain that evolve with velocities $\pm
v(k)=\pm\partial_k\varepsilon_{h\gamma}(k)$, and a non-zero connected
correlation can occur between two points in space-time only if a
quasi-particle can propagate from one to the other. Let us fix
$t_1<t_2$ and set $\delta=t_2-t_1$. According to this quasi-particle
picture, the operator $\sigma^x_{\ell+1}(t_1)$ can be considered as a
local operator with support in
$[\ell+1-v_{\max}t_1,\ell+1+v_{\max}t_1]$, with $v_{\max}=\max_k
|v(k)|$.
The condition for the connected correlator to be
non-negligible is for the supports of $\sigma^x_{\ell+1}(t_1)$ and
$\sigma^x_{1}(t_2)$ (or equivalently their backward light-cones) to
overlap. This explains why for $t_1+t_2<\tfrac{\ell}{v_{\max}}$
the connected $2$-point function is negligibly small. This corresponds
to the triangular blue regions in the bottom left corners of
Figs~\ref{ctdensity4}, which grow with $\ell$. On the other hand, we
expect on physical grounds that the effects of making a local
perturbation at time $t_1$ will become increasingly difficult to
detect if we wait long enough and connected correlations should
therefore decay with respect to the time  difference $|t_2-t_1|$ when
the latter gets large. This explains the smallness
of the connected two-point functions observed in the upper left and
bottom right corners of \ref{ctdensity4}.

\begin{figure}[H]
\begin{center}
 \includegraphics[scale=0.19]{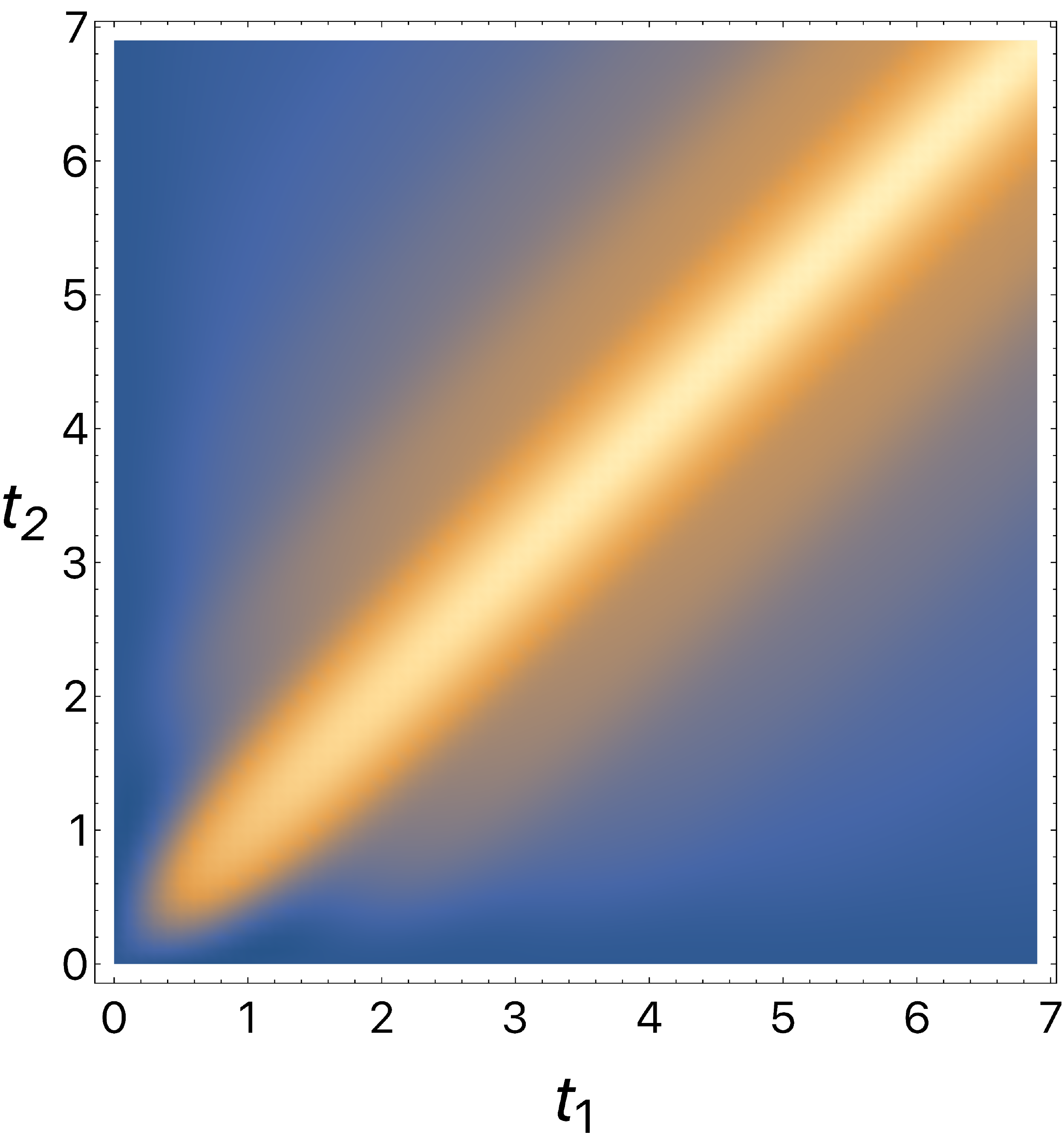}
 \includegraphics[scale=0.19]{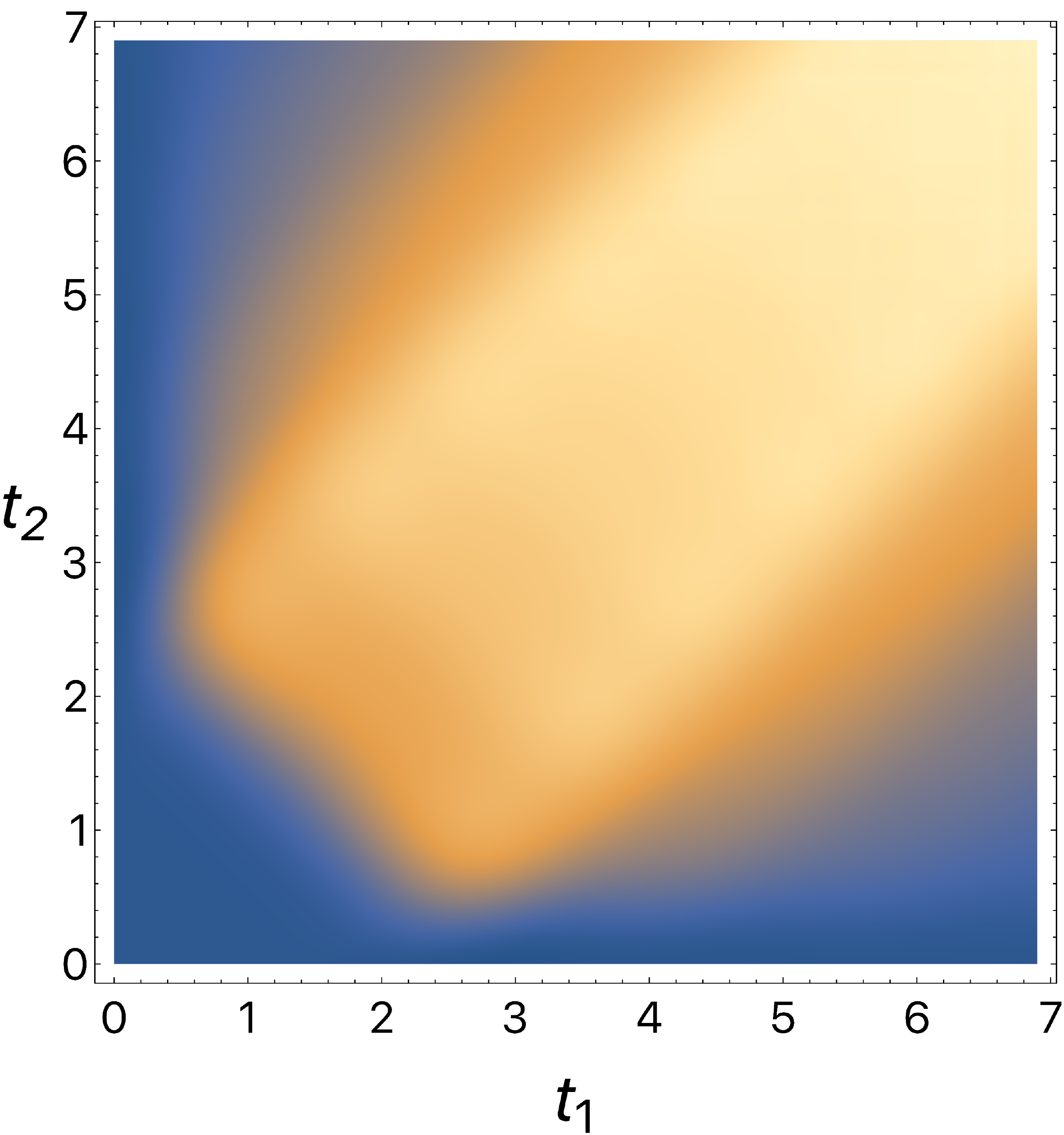}
 \includegraphics[scale=0.19]{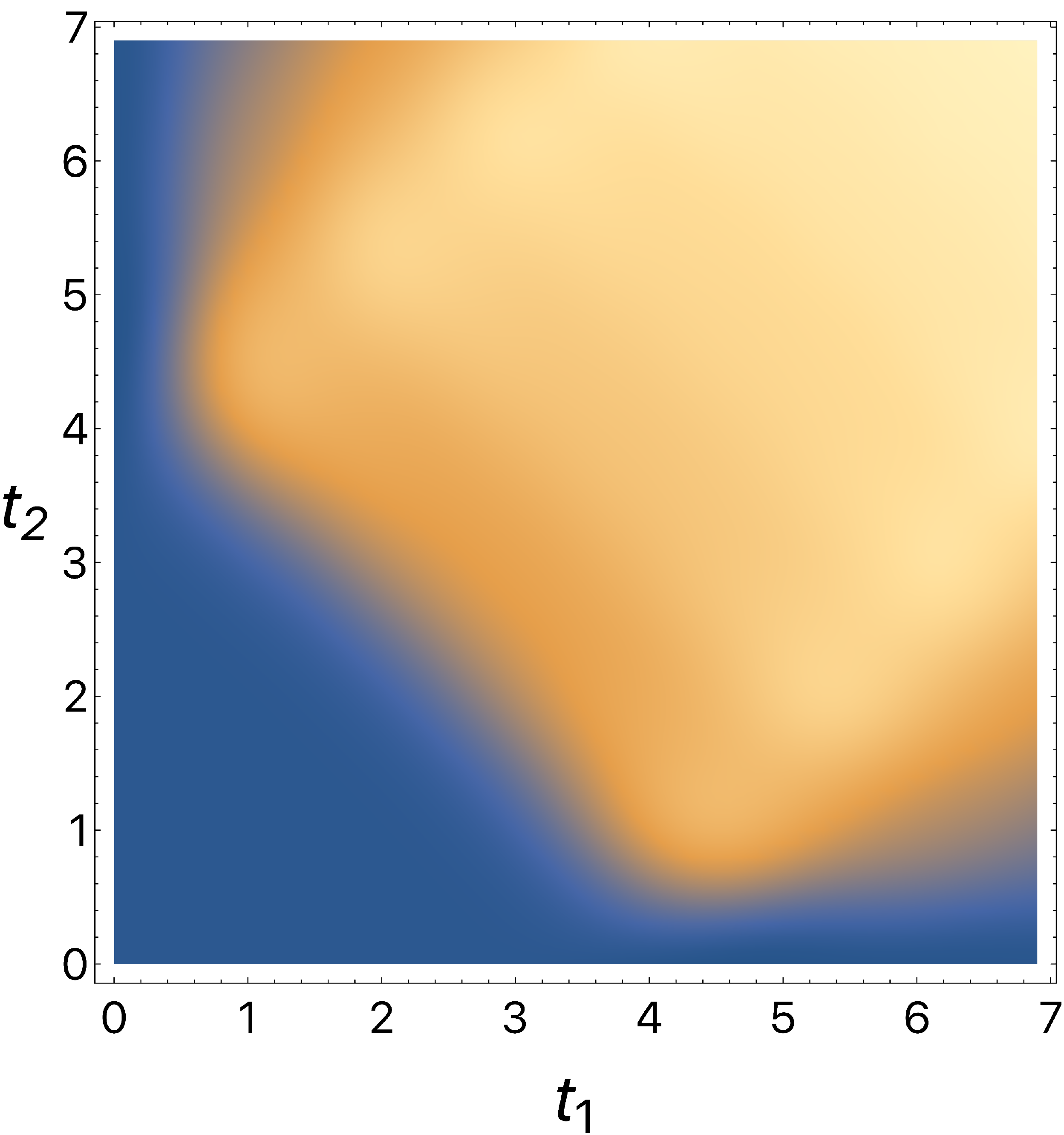}
\caption{Density plots of the connected non-equal time correlation
$\langle \sigma^x_{\ell+1}(t_2)\sigma^x_1(t_1)\rangle_c$ after a
quantum quench as a function of $t_1$ and $t_2$, for
$\ell=0,3,6$ from left to right. The three plots use different color
scales.}    
\label {ctdensity4}
\end{center}
\end {figure}

\section{Discussion and outlook}
In this work we have addressed the problem of computing
out-of-equilibrium observables in the XY spin chain subject to
arbitrary time variations of the magnetic field $h(t)$ and anisotropy
$\gamma(t)$. We obtained closed-form expressions for the
thermodynamic limit of the order parameter expectation value,
dynamical two-point function and static three-point function, as well
as of the full counting statistics of the transverse
magnetization. These expressions are valid for all times 
and for arbitrary distances. They hold not
only for out-of-equilibrium situations, but also in the wider context
of the Coherent Ensemble as introduced in the text. We emphasize that
to the best of our knowledge no exact explicit expressions in the
thermodynamic limit were known for the expectation value of operators
that are non-local in the underlying fermions, namely for the order
parameter one-point function, two-point dynamical correlation and
three-point static correlation. While the expectation values of
operators that are local in the underlying fermions can be
straightforwardly computed using Wick's theorem and do not require the
Fredholm determinant expressions derived in this paper, our method
provides a unified approach based on form factor summation, hence
better suited for generalization to the interacting case where no
Wick's theorem holds. Despite the free nature of the XY model, the
problem of performing the spectral sum over form factors was
previously solved only for free models with $U(1)$ symmetry such as
the impenetrable Bose gas or the XX chain.

In our derivation of these results we have followed a
different route than the ones traditionally used in the
computation of out-of-equilibrium dynamics in integrable
models. Our approach relies on remarkable properties of \textit{coherent
  states}, that are weighted superpositions (in a precise manner) of
exponentially many eigenstates of the Hamiltonian, that in a sense
behave more smoothly than pure eigenstates and are easier to
manipulate. Crucially they stay coherent when expressed in terms of
the eigenstates of the Hamiltonian at other values of $h$ and
$\gamma$, which allows one to carry out the calculations in a
preferred simple basis, such as $h=0$ and $\gamma=1$, where the form
factors are exactly Cauchy determinants. Efficient summation formulas
exploiting both the coherent state structure and the form factor
determinant structure eventually lead to our results. 

Our work opens up a number of future directions. 
The first direction is to determine the asymptotic behaviour of
the various correlation functions considered here. A second direction
is to investigate whether some ideas of this fruitful approach
can be generalized to an interacting case. Although the coherent
state structure used in this paper is rather fragile, there could be
analogous macroscopic superpositions of eigenstates in an interacting
model that enjoy similar interesting properties.

\appendix
\section{Diagonalization of the XY model in a field \label{xy}}
\subsection{Mapping to free fermions}
In this appendix we review how to diagonalize the XY
Hamiltonian with magnetic field $h$ and anisotropy $\gamma$ 
\be
H(h,\gamma)=-\sum_{j=1}^{L}\frac{1+\gamma}{2}\sigma_j^x\sigma_{j+1}^x
+\frac{1-\gamma}{2}\sigma_j^y\sigma_{j+1}^y+h\sigma_j^z\, ,
\ee 
where $\sigma_j^\alpha$ are the Pauli matrices at site $j$ and
\be
\sigma^\alpha_{L+1}=\sigma^\alpha_1\ ,\quad \alpha=x,y,z.
\ee
The quantum XY chain is mapped to a model of
spinless fermions by means of a Jordan-Wigner transformation. Defining
$\sigma^\pm_j=\big(\sigma^x_\ell\pm i\sigma^y_j\big)/2$ we construct
spinless fermion creation and annihilation operators by
\be
c_l^\dag=\prod_{j=1}^{l-1}\sigma_j^z \sigma_l^-\ ,\quad
\{c_j,c^\dagger_l\}=\delta_{j,l}.
\ee
The inverse transformation is
\begin{align}
\sigma^z_j=1-2c^\dagger_jc_j\ ,\
\sigma^x_j=\prod_{l=1}^{j-1}(1-2c^\dagger_lc_l)(c_j+c^\dagger_j)\ ,\
\sigma^y_j=i\prod_{l=1}^{j-1}(1-2c^\dagger_lc_l)(c^\dagger_j-c_j)\ .
\end{align}
The Hamiltonian can be expressed in terms of the fermions as
\begin{align}
H(h,\gamma)=&-\sum_{j=1}^{L-1}\frac{1+\gamma}{2}\big[c^\dagger_j-c_j\big]\big[c_{j+1}+c^\dagger_{j+1}\big]
-\sum_{j=1}^{L-1}\frac{1-\gamma}{2}\big[c^\dagger_j+c_j\big]\big[c_{j+1}-c^\dagger_{j+1}\big]\nn
&-h\sum_{j=1}^L[1-2c^\dagger_jc_j]\nn
&-e^{i\pi {\hat
    N}}\frac{1+\gamma}{2}(c_L-c_L^\dagger)(c_1+c_1^\dagger)-e^{i\pi {\hat N}}\frac{1-\gamma}{2}(c_L+c_L^\dagger)(c_1^\dagger-c_1)\,,
\end{align}
where
\be
\hat{N}=\sum_{j=1}^Lc^\dagger_j c_j.
\label{fermionnumber}
\ee
As $[H,e^{i\pi {\hat N}}]=0$ we may diagonalize the two operators
simultaneously. The Hamiltonian is thus block diagonal $H=H^{\rm
  NS}\oplus H^{\rm R}$, where $H^{\rm NS,R}$ act on the subspaces of
the Fock space with an even/odd number of fermions respectively. 
\subsection{Even fermion number}
In the sector with an even number of fermions we have
$e^{i\pi{\hat{N}}}=1$ and the Hamiltonian can be written in the form 
\begin{align}
H^{\rm
  NS}(h,\gamma)=&-\sum_{j=1}^{L}\frac{1+\gamma}{2}\big[c^\dagger_j-c_j\big]\big[c_{j+1}+c^\dagger_{j+1}\big]-\sum_{j=1}^{L}\frac{1-\gamma}{2}\big[c^\dagger_j+c_j\big]\big[c_{j+1}-c^\dagger_{j+1}\big]\nn 
&-h\sum_{j=1}^L[1-2c^\dagger_jc_j]
\label{Heven}
\end{align}
where we have imposed antiperiodic boundary conditions on the fermions
\be
c_{L+1}=-c_1.
\label{boundary conditions}
\ee
The Hamiltonian $H^{\rm NS}$ is diagonalized by going to Fourier space
\be
c(k_n)=\frac{1}{\sqrt{L}}\sum_{j=1}^Lc_j\ e^{ik_n j},
\ee
where $k_n$ are quantized according to (\ref{boundary conditions})
\be
k_{n}=\frac{2\pi (n+1/2)}{L}\ ,\quad n=-\frac{L}{2},\ldots \frac{L}{2}-1.
\label{NS}
\ee
The antiperiodic sector is commonly referred to as Neveu-Schwarz (NS)
sector. 
Introducing Bogoliubov fermions by
\begin{align}
c(k_n)&=\cos(\theta_{k_n}/2) \alpha_{h\gamma;k_n}
+ i\sin(\theta_{k_n}/2)\alpha_{h\gamma;-k_n}^\dagger\ ,\nn
c^\dagger(-k_n)&=i\sin(\theta_{k_n}/2) \alpha_{h\gamma;k_n}
+\cos(\theta_{k_n}/2)\alpha_{h\gamma;-k_n}^\dagger,
\label{Bogoliubovtrafo}
\end{align}
where the Bogoliubov angle fulfils
\begin{align}
\tan\theta_k=\left[\frac{\gamma\sin(k)}{\cos(k)-h}\right],
\label{angle<}
\end{align}
the Hamiltonian becomes diagonal
\be
H^{\rm NS}(h,\gamma)=\sum_{k\in{\rm NS}}
\varepsilon_{h\gamma}(k)\left[\alpha^\dagger_{h\gamma;k}\alpha_{h\gamma;k}-\frac{1}{2}\right].
\ee
Here the dispersion relation is
\begin{align}
\varepsilon_{h\gamma}(k)&=2\sqrt{(h-\cos k)^2+\gamma^2\sin^2 k}.
\end{align}
A basis for the Fock space in the sector with even fermion number is
then given by
\begin{align}
|k_1,\ldots,k_{2m}\rangle_{h\gamma}^{\rm NS}&=\prod_{j=1}^{2m}
\alpha^\dagger_{h\gamma;k_{j}}|0\rangle_{h\gamma}^{\rm NS}\ ,\quad
k_j\in{\rm NS},
\end{align}
where the fermion vacuum $|0\rangle_{h\gamma}^{\rm NS}$ is the state annihilated
by all $\alpha_{h\gamma;k_j}$ ($j=-\frac{L}{2},\ldots,\frac{L}{2}-1$).
\subsection{Odd fermion number}
In the sector with an odd number of fermions we have
$e^{i\pi{\hat{N}}}=-1$. The Hamiltonian can again be written in the form
\begin{align}
H^{\rm R}(h,\gamma)=&-\sum_{j=1}^{L}\frac{1+\gamma}{2}\big[c^\dagger_j-c_j\big]\big[c_{j+1}+c^\dagger_{j+1}\big]
-\sum_{j=1}^{L}\frac{1-\gamma}{2}\big[c^\dagger_j+c_j\big]\big[c_{j+1}-c^\dagger_{j+1}\big]\nn
&-h\sum_{j=1}^L[1-2c^\dagger_jc_j]
\label{Hodd}
\end{align}
but now we have to impose periodic boundary conditions on the fermions
\be
c_{L+1}=c_1.
\label{periodic boundary conditions}
\ee
In Fourier space we therefore now have
\be
c(p_n)=\frac{1}{\sqrt{L}}\sum_{j=1}^Lc_j\ e^{ip_n j},
\ee
where $p_n$ are quantized according to (\ref{periodic boundary conditions})
\be
p_{n}=\frac{2\pi n}{L}\ ,\quad n=-\frac{L}{2},\ldots \frac{L}{2}-1.
\label{R}
\ee
The periodic sector is known as Ramond sector. Defining Bogoliubov fermions
${\alpha}_{p_n}$ for $p_n\neq 0$ by 
\begin{align}
c(p_n)&=\cos(\theta_{p_n}/2) \alpha_{h\gamma;p_n}
+ i\sin(\theta_{p_n}/2)\alpha_{h\gamma;-p_n}^\dagger\ ,\nn
c^\dagger(-p_n)&=i\sin(\theta_{p_n}/2) \alpha_{h\gamma;p_n}
+\cos(\theta_{p_n}/2)\alpha_{h\gamma;-p_n}^\dagger,
\label{Bogoliubovtrafo2}
\end{align}
we can express the Hamiltonian as 
\begin{align}
H^{\rm R}(h,\gamma)&=\sum_{\ontop{p\in{\rm R}}{p\neq0}}
\varepsilon_{h\gamma}(p)\left[{\alpha}^\dagger_{h\gamma;p}{\alpha}_{h\gamma;p}-\frac{1}{2}\right]-2(1-h)\left[{\alpha}^\dagger_{h\gamma;0}
{\alpha}_{h\gamma;0}-\frac{1}{2}\right].
\end{align}
A basis of the subspace of the Fock space with odd fermion numbers is
then given by
\begin{align}
|p_1,\ldots,p_{2m+1}\rangle_{h\gamma}^{\rm R}&=\prod_{j=1}^{2m+1}
\alpha^\dagger_{h\gamma;p_{j}}|0\rangle_{h\gamma}^{\rm R}\ ,\quad p_j\in {\rm R},
\end{align}
where the fermion vacuum $|0\rangle_{h\gamma}^{\rm R}$ is the state annihilated by all
$\alpha_{h\gamma;p_j}$ ($j=-\frac{L}{2},\ldots,\frac{L}{2}-1$).

\section{Useful lemmas\label{lemmas}}
\subsection{Overlap and form factors}
\begin{property}\label{overlap}
Let $|\pmb{k}\rangle^{\rm NS}_{h\gamma}$ and $|\pmb{q}\rangle^{\rm
  NS}_{\tilde{h}\tilde{\gamma}}$ two eigenstates of the XY Hamiltonian
at different magnetic fields $h,\tilde{h}$ and anisotropies
$\gamma,\tilde{\gamma}$. 
We define $\pi(\pmb{k})$ as the subset of strictly positive elements
$k_n\in\pmb{k}$ such that $-k_n\in\pmb{k}$ as well, and
$\sigma(\pmb{k})$ the subset of unpaired momenta, i.e. $k_j\in\pmb{k}$
  but $-k_j\notin\pmb{k}$. $\pi(\pmb{q})$ and $\sigma(\pmb{q})$ are
  defined analogously.
Then the following formula for the overlap between the two states holds
\begin{equation}
{}^{\rm NS}_{\tilde{h}\tilde{\gamma}}\langle \pmb{q}|\pmb{k}\rangle^{\rm NS}_{h\gamma}=\1_{\sigma(\pmb{q})=\sigma(\pmb{k})}\frac{\prod_{p\in \sigma(\pmb{k})}\frac{1}{\cos (\theta_k^{\tilde{h}\tilde{\gamma}}-\theta_k^{h\gamma})/2}\prod_{p\in \pi(\pmb{q})\perp \pi(\pmb{k})}iK_{\tilde{h}\tilde{\gamma};h\gamma}(p)}{\prod_{p\in{\rm NS}_+}\sqrt{1+K^2_{\tilde{h}\tilde{\gamma};h\gamma}(p)}}\,,
\end{equation}
where we defined $ \pi(\pmb{q})\perp \pi(\pmb{k})=\pi(\pmb{q})\cup \pi(\pmb{k})-(\pi(\pmb{q})\cap\pi(\pmb{k}))$. 
\end{property}
\begin{proof}
As shown in Appendix \ref{xy}, we have the following relation
between Bogoliubov fermion operators at different values of
magnetic fields and anisotropies 
\begin{equation}
\alpha_{h\gamma;k}= \cos \frac{\theta_k^{\tilde{h}\tilde{\gamma}}-\theta_k^{h\gamma}}{2}\alpha_{\tilde{h}\tilde{\gamma};k}+i\sin \frac{\theta_k^{\tilde{h}\tilde{\gamma}}-\theta_k^{h\gamma}}{2}\alpha^\dagger_{\tilde{h}\tilde{\gamma};-k}\,.
\end{equation}
From this, we deduce the relation between the vacuum states
\begin{equation}\label{vaccu}
|0\rangle_{h\gamma}^{\rm NS}=\prod_{p\in {\rm NS}_+}\left[\frac{1+iK_{\tilde{h}\tilde{\gamma};h\gamma}(p)\alpha_{\tilde{h}\tilde{\gamma};-p}^{\dagger}\alpha_{\tilde{h}\tilde{\gamma};p}^{\dagger}}{\sqrt{1+K^2_{\tilde{h}\tilde{\gamma};h\gamma}(p)}}\right]|0\rangle_{\tilde{h}\tilde{\gamma}}^{\rm NS}\,.
\end{equation}
Indeed, the right-hand side is annihilated by all the $\alpha_{h\gamma,k}$
From these relations one deduces the overlap given in the Lemma. 
\end{proof}

\begin{property}[Form factor of $e^{i\theta\sum_{j=1}^{\ell} \sigma_j^z}$\label{ff1}]
If $\pmb{\lambda}$ and $\pmb{\mu}$ have the same number of elements
 the following
  determinant representation holds
\begin{align}
{}^{\rm NS}_{\infty\gamma}\langle \pmb{\lambda}|e^{i\theta
  \sum_{j=1}^{\ell} \sigma^z_j}|\pmb{\mu}\rangle_{\infty\gamma}^{\rm
  NS}&=e^{i\theta\ell }\det E(\pmb{\lambda},\pmb{\mu})\ ,\nn
\label{E}
E(\pmb{\lambda},\pmb{\mu})_{jk}&=\begin{cases}
\frac{e^{-2i\theta}-1}{L}e^{i(\lambda_j-\mu_k)}\frac{1-e^{i\ell(\lambda_j-\mu_k)}}{1-e^{i(\lambda_j-\mu_k)}}
& \text{if }\lambda_j\neq\mu_k\ ,\\
1+\tfrac{\ell}{L}(e^{-2i\theta}-1)& \text{if }\lambda_j=\mu_k\ .
\end{cases}
\end{align}
If they have different numbers of elements the form
factor vanishes. 
\end{property}
\begin{proof}
Since the operator conserves the number of particles, $\pmb{\lambda}$
and $\pmb{\mu}$ must have the same number of particles for the form
factor not to vanish. Let us denote this number by $N$. Using
that at $h=\infty$ the Bogoliubov fermions reduce to the Jordan-Wigner
fermions we have 
\begin{align}
{}^{\rm NS}_{\infty\gamma}\langle \pmb{\lambda}|e^{i\theta \sum_{j=1}^{\ell} \sigma^z_j}|\pmb{\mu}\rangle_{\infty\gamma}^{\rm NS}&=\frac{1}{L^N}\sum_{j_1...j_N}\sum_{k_1...k_N}{}^{\rm NS}_{\infty\gamma}\langle 0|c_{j_N}...c_{j_1}e^{i\theta \sum_{j=1}^{\ell} \sigma^z_j}c_{k_1}^\dagger...c_{k_N}^\dagger |0\rangle_{\infty\gamma}^{\rm NS} \nn
&\qquad\qquad\qquad\qquad\times e^{-ik_1\mu_1-...-ik_N\mu_N}e^{ij_1\lambda_1+...+ij_N\lambda_N}\nn
&=\frac{e^{i\theta\ell}}{L^N}\!\sum_{\sigma\in\mathfrak{S}_N}(-1)^\sigma\!\sum_{j_1...j_N}e^{ij_1(\lambda_1-\mu_{\sigma(1)})+...+ij_N(\lambda_N-\mu_{\sigma(N)})}e^{-2i\theta \sum_{q=1}^N \1_{j_q\leq\ell}}\nn
&=\frac{e^{i\theta\ell}}{L^N}\sum_{\sigma\in\mathfrak{S}_N}(-1)^\sigma\prod_{q=1}^N\sum_{j=1}^{L}e^{ij(\lambda_q-\mu_{\sigma(q)})}e^{-2i\theta \1_{j\leq\ell}}\nn
&=e^{i\theta\ell}\det E(\pmb{\lambda},\pmb{\mu})\,.
\end{align}
\end{proof}

\begin{property}[Form factors of $\sigma^x$\label{ffsigmax}]
The form factors of $\sigma^x_\ell$ between energy eigenstates at
$h=0$, $\gamma=1$ have the following determinant representation
\begin{equation}
\begin{aligned}
{}^{\rm R}_{01}\langle \pmb{\lambda}\cup\{0\}|\sigma^x_\ell|\pmb{\mu}\rangle^{\rm NS}_{01}=&(-1)^{N(N+1)/2}\left(\frac{2}{L}\right)^Ne^{\tfrac{i}{2}(\sum_{\lambda\in\pmb{\lambda}}\lambda+\sum_{\mu\in\pmb{\mu}}\mu)}e^{-i\ell(\sum_{\lambda\in\pmb{\lambda}}\lambda-\sum_{\mu\in\pmb{\mu}}\mu)}\\
&\times\det\left[\frac{1}{e^{i\lambda_j}-e^{i\mu_k}} \right]_{jk}\ .
\end{aligned}
\end{equation}
Here $N$ is the number of elements in $\pmb{\lambda}$ and
$\pmb{\mu}$. If they have different numbers of elements the form
factor vanishes. 
\end{property}
\begin{proof}
See Refs\cite {Bugrij,BL03,Gehlen,iorgov11}. 
\end{proof}

\begin{property}[Form factors of $e^{-itH(h,\gamma)}$\label{ffh}]
The form factors of $e^{-itH}$ between energy eigenstates at
$h=0$, $\gamma=1$ have the following representation
\begin{equation}\label{ffheq}
\begin{aligned}
{}^{\rm NS}_{01}\langle \pmb{\lambda}|e^{-itH(h,\gamma)}|\pmb{\mu}\rangle^{\rm NS}_{01}&=\1_{\sigma(\pmb{\lambda})=\sigma(\pmb{\mu})}e^{-it\mathfrak{E}^{\rm NS}}\prod_{k\in{\rm NS}_+}\frac{1+K^2(k)e^{-2it\varepsilon(k)}}{1+K^2(k)}\\
&\times\prod_{\nu\in\sigma(\pmb{\lambda})}e^{-it\varepsilon(\nu)}\frac{1+K^2(\nu)}{1+e^{-2it\varepsilon(\nu)}K^2(\nu)}\\
&\times \prod_{\lambda\in\pi(\pmb{\lambda})}iK(\lambda)\frac{1-e^{-2it\varepsilon(\lambda)}}{1+K^2(\lambda)e^{-2it\varepsilon(\lambda)}}\prod_{\mu\in\pi(\pmb{\mu})}(-iK(\mu))\frac{1-e^{-2it\varepsilon(\mu)}}{1+K^2(\mu)e^{-2it\varepsilon(\mu)}}\\
&\times \prod_{\nu\in\pi(\pmb{\lambda})\cap\pi(\pmb{\mu})}
\frac{(1+e^{-2it\varepsilon(\nu)}K^2(\nu))(1+e^{-2it\varepsilon(\nu)}/K^2(\nu))}{(1-e^{-2it\varepsilon(\nu)})^2}\ .
\end{aligned}
\end{equation}
Here the notations are as in Lemma \ref{overlap} and we have
  used shorthand notations $K(k)=K_{01;h\gamma}(k)$, $\varepsilon(k)=\varepsilon_{h\gamma}(k)$. 
\end{property}
\begin{proof}
Inserting two resolutions of the identity in terms of energy
  eigenstates on either side of $e^{-iH(h,\gamma)t}$ we obtain
\begin{align}
&{}^{\rm NS}_{01}\langle\pmb{\lambda}|e^{-itH(h,\gamma)}|\pmb{\mu}\rangle^{\rm NS}_{01}=\sum_{\pmb{\nu}}{}^{\rm NS}_{01}\langle\pmb{\lambda}|\pmb{\nu}\rangle^{\rm NS}_{h\gamma}\ {}^{\rm NS}_{h\gamma}\langle \pmb{\nu}|\pmb{\mu}\rangle^{\rm NS}_{01} e^{-it\mathfrak{E}^{\rm NS}}\prod_{\nu\in\pmb{\nu}}e^{-it\varepsilon(\nu)}\nn
&=1_{\sigma(\pmb{\lambda})=\sigma(\pmb{\mu})}e^{-it\mathfrak{E}^{\rm
      NS}}\prod_{k\in
    \sigma(\pmb{\lambda})}e^{-it\varepsilon(k)}(1+K^2(k))\prod_{\mu\in\pi(\pmb{\mu})}(-iK(\mu))
\prod_{\lambda\in\pi(\pmb{\lambda})}iK(\lambda)\prod_{p\in{\rm NS}_+}\frac{1}{1+K^2(p)}\nn
&\qquad\qquad\qquad\times\sum_{\substack{\pmb{\nu}\subset {\rm
      NS}_+\\ \cap
    [\sigma(\pmb{\lambda})\cup(-\sigma(\pmb{\lambda}))]=\emptyset}}\prod_{\nu\in\pmb{\nu}}e^{-2it\varepsilon(\nu)}\begin{cases}\frac{1}{K^2(\nu)}&
  \text{if }\nu\in\pi(\pmb{\lambda})\cap \pi(\pmb{\mu})\ ,\\ -1& \text{if
  }\nu\in\pi(\pmb{\lambda})\perp \pi(\pmb{\mu})\ ,\\ K^2(\nu)& \text{if
  }\nu\notin\pi(\pmb{\lambda})\cup\pi(\pmb{\mu})\ .\end{cases}
\label{lemma4_1}
\end{align}
The last line can be rewritten in the form
\begin{align}
\prod_{\substack{k\in \pi(\pmb{\lambda})\cap \pi(\pmb{\mu})}}[1+\tfrac{e^{-2it\varepsilon(k)}}{K^2(k)}]\prod_{\substack{k\in \pi(\pmb{\lambda})\perp \pi(\pmb{\mu})}}[1-e^{-2it\varepsilon(k)}]\prod_{\substack{k\notin \pi(\pmb{\lambda}), \pi(\pmb{\mu})\\ \notin [\sigma(\pmb{\lambda})\cup(-\sigma(\pmb{\lambda}))]}}[1+e^{-2it\varepsilon(k)}K^2(k)]\,.
\end{align}
Substituting this back in \fr{lemma4_1} results in the
  representation given in the Lemma.
\end{proof}

\subsection{Summation formulas}
\begin{property}[Andr\'eief identity \cite{andreief}]\label{andreief}
Given two functions $f(\lambda,\mu)$ and $g(\lambda,\mu)$, a set $K$ and two sets of numbers $\{\lambda_i\}_{i=1}^N,\{\mu_j\}_{j=1}^N$ we have the relation
\begin{equation}
\sum_{k_1<...<k_N\in K}\det_{i,j} \left[f(\lambda_i,k_j) \right]\det_{i,j} \left[g(k_i,\mu_j) \right]=\det_{i,j} \left[\sum_{k\in K} f(\lambda_i,k)g(k,\mu_j) \right]\,.
\end{equation}
\end{property}
\begin{proof}
See Ref. \cite{DBG}.
\end{proof}

\begin{property}[de Bruijn identity \cite{debrujin}]\label{brujin}
Let $f(\lambda,\mu)$ and $g(\lambda,\mu)$ be two functions of two
  variables, a set $K$ and a set of numbers
  $\{\lambda_i\}_{i=1}^{2N}$. Define a matrix
\be
\Big(A(\pmb{k})\Big)_{ij}=\begin{cases} 
f(k_q,\lambda_j)& \text{if }i=2q-1\ ,\\
g(k_{q},\lambda_j)&  \text{if }i=2q\ .
\end{cases}
\ee
Then the following indentity holds
\begin{equation}\label{brujin2}
\sum_{k_1<...<k_N\in K}\det A(\pmb{k})=\underset{ij}{{\rm pf}} \left[ \sum_{k\in K} f(k,\lambda_i)g(k,\lambda_j)-f(k,\lambda_j)g(k,\lambda_i)\right]\,.
\end{equation}
\end{property}

\begin{proof}
Using the definition of the determinant we have 
\begin{equation}
\det A(\pmb{k})=\sum_{\sigma\in\mathfrak{S}_{2N}}(-1)^\sigma f(k_1,\mu_{\sigma (1)})g(k_1,\mu_{\sigma (2)})...f(k_N,\mu_{\sigma(2N-1)})g(k_{N},\mu_{\sigma(2N)})\,.
\end{equation}
Then
\begin{align}
\sum_{k_1<...<k_N}\det A(\pmb{k})&=\frac{1}{N!}\sum_{k_1,...,k_N}\det A(\pmb{k})
=\frac{1}{N!}\sum_{\sigma\in\mathfrak{S}_{2N}}(-1)^\sigma b_{\sigma(1)\sigma(2)}...b_{\sigma(2N-1)\sigma(2N)}\,,
\end{align}
where
\begin{equation}
b_{ij}=\sum_k f(k,\mu_i)g(k,\mu_j)\,.
\end{equation}
Changing variables to $\sigma=\sigma' \cdot (1,N+1)$ we have
\begin{align}
\sum_{\sigma\in\mathfrak{S}_{2N}}(-1)^\sigma
b_{\sigma(1)\sigma(2)}...b_{\sigma(2N-1)\sigma(2N)}&=
\sum_{\sigma\in\mathfrak{S}_{2N}}\frac{(-1)^\sigma}{2} (b_{\sigma(1)\sigma(2)}-b_{\sigma(2)\sigma(1)})...b_{\sigma(2N-1)\sigma(2N)}\nn
&=\frac{1}{2^N}\sum_{\sigma\in\mathfrak{S}_{2N}}(-1)^\sigma B_{\sigma(1)\sigma(2)}...B_{\sigma(2N-1)\sigma(2N)}\,,
\end{align}
where $B_{ij}$ is the matrix on the right-hand side in the
Lemma. This completes the proof.
\end{proof}

\subsection{Coherent averages}
\begin{property}\label{qa}
Let $F[\pmb{q}]$ be a function of $\pmb{q}$, and $f(k)$ a function. We define
\begin{equation}
\langle F\rangle\equiv\frac{1}{\prod_{k\in{\rm NS}_+}[1+|f(k)|^2]}\sum_{\pmb{q}\subset{\rm NS}_+}F[\pmb{q}]\prod_{q\in\pmb{q}}[|f(q)|^2]\,.
\end{equation}
If in the thermodynamic limit $F[\pmb{q}]$ depends on the momenta only
through the root density $\rho$, i.e. $\lim_{\rm
  th}F[\pmb{q}]=F[\rho]$, then 
\begin{align}
\langle F\rangle&=F[\rho_s]+o(L^{0})\ ,\nn
\rho_s(k)&=\frac{1}{2\pi}\frac{|f(k)|^2}{1+|f(k)|^2}\ .
\end{align}
\end{property}
\begin{proof}
See Ref. \cite{DBG}.
\end{proof}

\begin{property}\label{qa2}
Given two functionals $F[\pmb{q}]$ and $G[\pmb{q}]$, as well as three functions $f(k),g(k),h(k)$, we define
\begin{equation}\label{deno}
\langle F,G\rangle\equiv\frac{\sum_{\pmb{\lambda},\pmb{\mu}\subset {\rm NS}_+}F[\pmb{\lambda}]G[\pmb{\mu}] \prod_{\lambda\in\pmb{\lambda}}f(\lambda)\prod_{\mu\in\pmb{\mu}}g(\mu)\prod_{\nu\in\pmb{\lambda}\cap \pmb{\mu}}h(\nu)}{\prod_{k\in {\rm NS}_+}[1+f(k)+g(k)+f(k)g(k)h(k)]}\,.
\end{equation}
If $F[\pmb{q}]$ and $G[\pmb{q}]$ depend only on the root density of $\pmb{q}$ in the thermodynamic limit, then
\begin{equation}
\langle F,G\rangle=F[\rho_1]G[\rho_2]+o(L^{-0})\,,
\end{equation}
with
\begin{equation}
\rho_1=\frac{1}{2\pi}\frac{f+fgh}{1+f+g+fgh}\,,\qquad \rho_2=\frac{1}{2\pi}\frac{g+fgh}{1+f+g+fgh}\,.
\end{equation}
In these equations the root density can be complex.
\end{property}
\begin{proof}
It is a generalisation of the proof of Lemma \ref{qa} in \cite{DBG}. Let us first treat the particular case where in the thermodynamic limit $F$ and $G$ depend only on $r$ the number of elements of $\pmb{q}$ divided by $L$. We introduce the generating function
\begin{equation}\label{alpha}
    \Gamma(\alpha;\beta)=\frac{\sum_{\pmb{\lambda},\pmb{\mu}\subset {\rm NS}_+}\prod_{\lambda\in\pmb{\lambda}}[1+\tfrac{\alpha}{L}]f(\lambda)\prod_{\mu\in\pmb{\mu}}[1+\tfrac{\beta}{L}]g(\mu)\prod_{\nu\in\pmb{\lambda}\cap \pmb{\mu}}h(\nu)}{\prod_{k\in {\rm NS}_+}[1+f(k)+g(k)+f(k)g(k)h(k)]}\,.
\end{equation}
We note that the denominator is such that $\Gamma(0;0)=1$. By differentiating with respect to $\alpha$ and $\beta$, we see that
\begin{equation}
    \langle r^i,r^j\rangle=\partial_\alpha^i \partial_\beta^j \Gamma(0;0)+\mathcal{O}(L^{-1})\,.
\end{equation}
Besides, performing the summation on $\pmb{\lambda},\pmb{\mu}$ we obtain
\begin{equation}
    \Gamma(\alpha)=\prod_{k\in{\rm NS}_+}\left[1+\frac{\alpha}{L} \frac{f+fgh}{1+f+g+fgh}+\frac{\beta}{L} \frac{g+fgh}{1+f+g+fgh}+\frac{\alpha\beta}{L^2}\frac{fgh}{1+f+g+fgh} \right](k)\,.
\end{equation}
From this we find for any $i,j$
\begin{equation}
  \langle r^i,r^j\rangle=\left(\int_0^\pi \rho_1(k)\D{k} \right)^i\left(\int_0^\pi \rho_2(k)\D{k} \right)^j+\mathcal{O}(L^{-1})\,,
\end{equation}
with $\rho_1,\rho_2$ given in the Lemma. As any regular function can be approximated by a polynomial with arbitrary precision provided its degree is high enough, this establishes the result of the Lemma when $F$ and $G$ are functions of $r$ only.

Let us now divide $[0,\pi]$ into $m$ windows $W_k=[\frac{\pi}{m}(k-1),\frac{\pi}{m}k]$ for $k=1,...,m$, and consider $F[r_1,...,r_m]$, $G[r_1,...,r_m]$ functions of $\pmb{q}$ that  in the thermodynamic limit depend only on $r_k$'s, the number of elements of $\pmb{q}$ in $W_k$ divided by $L$. By introducing $\Gamma(\alpha_1,...,\alpha_m;\beta_1,...,\beta_m)$ as in \eqref{alpha} with $\alpha$ replaced by $\alpha_k$ where $k$ is such that $\lambda,\mu\in W_k$, we get similarly
\begin{equation}
  \langle r_1^{i_1}...r_m^{i_m},r_1^{j_1}...r_m^{j_m}\rangle=\prod_{a=1}^m\left(\int_{ W_{a}} \rho_{1} \right)^{i_a}\left(\int_{ W_{a}} \rho_{2} \right)^{j_a}+\mathcal{O}(L^{-1})\,.
\end{equation}
Hence the Lemma holds whenever $F$, $G$ are functions of $r_1,...,r_m$ only. Since any regular functional of $\rho$ can be approximated with arbitrary precision by such a function provided $m$ is large enough, the Lemma holds for general $F[\rho]$ and $G[\rho]$.
\end{proof}

\subsection{Fredholm determinants}
\begin{property}[Generalized Cramer's rule]\label{cramer}
Let $A$ be an $N\times N$ matrix and $x^{i_1},...,x^{i_k}$ vectors of
size $N$. We define $B$ to be the matrix obtained from $A$ by replacing the
columns $i_1,...,i_k$ by $x^{i_1},...,x^{i_k}$. Then 
\begin{equation}\label{res}
\det B= \det Y \det A\,,
\end{equation}
with $Y$ the $k\times k$ matrix with entries
  $Y_{i_a,i_b}=y^{i_a}_{i_b}$, where the vector $y^{i_a}$ is a solution to 
\begin{equation}\label{system2}
Ay^{i_a}=x^{i_a}\,.
\end{equation}
\end{property}
\begin{proof}
Denoting the columns of $A$ by $C_1,...,C_{N}$ Eq
\eqref{system2} reads
\begin{equation}
x^{i_a}=\sum_{j=1}^N y^{i_a}_jC_j\,.
\end{equation}
Using the multilinearity of the determinant, one has
\begin{equation}
\det B=\sum_{j_1,...,j_k=1}^N y^{i_1}_{j_1}...y^{i_k}_{j_k} \det A^{j_1,...,j_k}\,,
\end{equation}
where $A^{j_1,...,j_k}$ denotes the matrix obtained from $A$ by
replacing the columns $i_1,...,i_k$ by $C_{j_1},...,C_{j_k}$. Since its
determinant is non-zero only if $\{j_1,...,j_k\}=\{i_1,...,i_k\}$, we
obtain 
\begin{equation}
\begin{aligned}
\det B&=\sum_{\sigma\in\mathfrak{S}_k} y^{i_1}_{i_{\sigma(1)}}...y^{i_k}_{i_{\sigma(k)}} \det A^{i_{\sigma(1)},...,i_{\sigma(k)}}\\
&=\det Y\det A\,.
\end{aligned}
\end{equation}
\end{proof}

\begin{property}\label{fredh}
Let $f(\lambda,\mu)$ be a function of two variables, $J\subset \{1,...,L\}$
a set of $n$ indices, $(a_{ij})_{i,j\in J}$ $n^2$ numbers and
$(g_i(\mu))_{i\in J}$, $(h_j(\lambda))_{j\in J}$ $2n$ functions of a
single variable. Define an $L\times L$ matrix $A$ by
\begin{equation}
A_{ij}=\begin{cases}
\delta_{ij}+\frac{1}{L}f(\tfrac{i}{L},\tfrac{j}{L})& \text{\rm if }i,j\notin J\ ,\\
\tfrac{1}{L}g_j(\tfrac{i}{L})& \text{\rm if }j\in J,\, i\notin J\ ,\\
\tfrac{1}{L}h_i(\tfrac{j}{L})& \text{\rm if }i\in J,\, j\notin J\ ,\\
\tfrac{1}{L}a_{ij}& \text{if }i,j\in J \ .
\end{cases}
\end{equation}
Then in the limit $L\to\infty$ the following Fredholm determinant
  representaion holds
\begin{equation}\label{reslem10}
\det A=\frac{1}{L^n}\det\left[a_{ij}-\int_0^1\int_0^1 h_i(\lambda)g_j(\mu)\phi(\lambda,\mu)\D{\lambda}\D{\mu}\right]_{1\leq i,j\leq n} \cdot \Det[{\rm Id}+f]+o(L^{-n})\,,
\end{equation}
where $\phi$ is the resolvent of the Fredholm equation
\begin{equation}\label{reso1}
\phi(\lambda,\mu)+\int_0^1 f(\lambda,\nu)\phi(\nu,\mu)\D{\nu}=\delta(\lambda-\mu)\,.
\end{equation}
\end{property}
\begin{proof}
Using Lemma \ref{cramer}, one has
\begin{equation}
\det A=\det A'\det X \ .
\end{equation}
Here $A'$ is the $L\times L$ matrix
\begin{equation}
A'_{ij}=\begin{cases}
\delta_{ij}+\frac{1}{L}f(\tfrac{i}{L},\tfrac{j}{L})& \text{if }i\notin
J\ ,\\
\tfrac{1}{L}h_i(\tfrac{j}{L})& \text{if }i\in J,\, j\notin J\ ,\\
\delta_{ij}+\frac{1}{L}f(\tfrac{i}{L},\tfrac{j}{L})& \text{if }i,j\in
J\ ,
\end{cases}
\end{equation}
and $X$ the $n\times n$ matrix whose entry $X_{ij}$ for $i,j\in J$ is
the $i$-th element of the solution $x^j$ to the linear system 
\begin{equation}
\label{system}
A'x^j=b^j\,,
\end{equation}
with the vector
\begin{equation}
b^j_i=\begin{cases}
\frac{1}{L}g_j(\tfrac{i}{L})& \text{if }i\notin J\ ,\\
\tfrac{1}{L}a_{ij}& \text{if }i\in J\ .
\end{cases}
\end{equation}
As $A'$ differs from the matrix with elements
$\delta_{ij}+\frac{1}{L}f(\tfrac{i}{L},\tfrac{j}{L})$ by only
off-diagonal terms on a finite number of rows in the thermodynamic
limit, one has the Fredholm determinant
\begin{equation}
\det A'=\Det[{\rm Id}+f]+o(L^0)\,.
\end{equation}
For $i\notin J$ the system \eqref{system} reads
\begin{equation}
x^j_i+\frac{1}{L}\sum_{k=1}^L f(\tfrac{i}{L},\tfrac{k}{L})x_k^j=\frac{1}{L}g_j(\tfrac{i}{L})\,.
\end{equation}
This equation allows one to describe the entries $x_i^j$ for $i\notin J$
by a function $x^j(\lambda)$ in the thermodynamic limit, since
$x_k^j$ for $k\in J$ appears a finite number of times with a factor
$\tfrac{1}{L}$. However, the entry $x_i^j$ for $i\in J$ can be
discontinuous in the $i$ direction in the thermodynamic limit. 
We thus obtain an integral equation for $x^j(\lambda)$ 
\begin{equation}
x^j(\lambda)+\int_0^1 f(\lambda,\nu)x^j(\nu)\D{\nu}=\frac{1}{L}g_j(\lambda)\,,
\end{equation}
whose solution is expressed as
\begin{equation}
x^j(\lambda)=\frac{1}{L}\int_0^1 \phi(\lambda,\mu)g_j(\mu)\D{\mu}\,,
\end{equation}
with the resolvent $\phi(\lambda,\mu)$ defined by
\begin{equation}
\phi(\lambda,\mu)+\int_0^1 f(\lambda,\nu)\phi(\nu,\mu)\D{\nu}=\delta(\lambda-\mu)\,,
\end{equation}
or equivalently
\begin{equation}\label{reso2}
\phi(\lambda,\mu)+\int_0^1 \phi(\lambda,\nu)f(\nu,\mu)\D{\nu}=\delta(\lambda-\mu)\,.
\end{equation}
Now, for $i,j\in J$ the system \eqref{system} reads
\begin{equation}
x^j_i+\frac{1}{L}\sum_{k=1}^L h_i(\tfrac{k}{L})x_k^j=\frac{1}{L}a_{ij}\,.
\end{equation}
In the thermodynamic limit, this yields for $i,j\in J$
\begin{equation}
x^j_i=\frac{1}{L}\left(a_{ij}-\int_0^1 \int_0^1 h_i(\lambda)g_j(\mu)\phi(\lambda,\mu)\D{\lambda}\D{\mu} \right)\,.
\end{equation}
This concludes the proof of the Lemma.

\end{proof}

\subsection{Miscellaneous}
\begin{property}\label{sum}
For $z$ complex and $0\leq \ell< L$ an integer, we have in finite size $L$
\begin{equation}\label{sum1}
\sum_{\nu\in{\rm R}}\frac{e^{i(\ell+1)\nu}}{e^{i\nu}-z}=L\frac{z^{\ell}}{1-z^L}\,,
\end{equation}
and
\begin{equation}\label{sum2}
\sum_{\nu\in{\rm NS}}\frac{e^{i(\ell+1)\nu}}{e^{i\nu}-z}=L\frac{z^{\ell}}{1+z^L}\,.
\end{equation}
\end{property}
\begin{proof}
Let us start with \eqref{sum1}. The left-hand side is a meromorphic function of $z$ with simple poles in $e^{i{\rm R}}$, i.e. for $z^L=1$. The full result can thus be obtained with analytic continuation from $|z|<1$. In this region, one has
\begin{equation}
\sum_{\nu\in{\rm R}}\frac{e^{i(\ell+1)\nu}}{e^{i\nu}-z}=\sum_{m\geq 0}z^m\sum_{\nu\in{\rm R}}e^{i(\ell-m)\nu}\,.
\end{equation}
If $\ell-m$ is a multiple of $L$, then the sum over $\nu$ is $L$, and otherwise the sum vanishes. Hence
\begin{equation}
\begin{aligned}
\sum_{\nu\in{\rm R}}\frac{e^{i(\ell+1)\nu}}{e^{i\nu}-z}&=\sum_{k\geq 0} z^{\ell+kL}L\\
&=L\frac{z^{\ell}}{1-z^L}\,.
\end{aligned}
\end{equation}
Now, to show \eqref{sum2} we write ${\rm NS}={\rm R}+\tfrac{\pi}{L}$ to obtain
\begin{equation}
\sum_{\nu\in{\rm NS}}\frac{e^{i(\ell+1)\nu}}{e^{i\nu}-z}=e^{i\ell \pi/L}\sum_{\nu\in{\rm R}}\frac{e^{i(\ell+1)\nu}}{e^{i\nu}-ze^{-i\pi/L}}\,.
\end{equation}
Using \eqref{sum1}, we then obtain \eqref{sum2}.
\end{proof}

\end{document}